\pdfoutput=1
\documentclass[11pt,twoside,a4paper,notitlepage]{article}
\def\USEBIBLATEX{} %
\def\ARTICLEFORMAT{}
\usepackage[a4paper,margin=2.5cm]{geometry}
\makeatletter{}%

        \usepackage[maxnames=10,backend=bibtex]{biblatex} %
        \renewbibmacro*{doi+eprint+url}{%
                \iftoggle{bbx:url} {\iffieldundef{doi}{\usebibmacro{url+urldate}}{}} {}%
                \newunit\newblock \iftoggle{bbx:eprint} {\iffieldundef{url}{\iffieldundef{doi}{\usebibmacro{eprint}}{}}{}} {}%
                \newunit\newblock \iftoggle{bbx:doi} {\printfield{doi}} {}}  

        \defbibheading{bibliography}[\refname]{}        
        \bibliography{literature.bib} 

        \usepackage{hyperref}

\usepackage[utf8]{inputenc}

\usepackage{color}

\usepackage{tikz}

\usepackage{subfigure}

\usepackage[noend]{algpseudocode}

\usepackage{xspace}

\usepackage{bm}

\usepackage{algorithm}
\ifthenelse{\isundefined{\ARTICLEFORMAT}}{

\numberofauthors{3} %
\author{
\alignauthor
Christoph Berkholz\\
       \affaddr{Humboldt-Universität zu Berlin}\\
       \email{berkholz@informatik.hu-berlin.de}
\alignauthor
Jens Keppeler\\
       \affaddr{Humboldt-Universität zu Berlin}\\
       \email{keppelej@informatik.hu-berlin.de}
\alignauthor Nicole Schweikardt\\
       \affaddr{Humboldt-Universität zu Berlin}\\
       \email{schweika@informatik.hu-berlin.de}
}

}{

  \author{Christoph Berkholz, Jens Keppeler, Nicole Schweikardt \\
    Humboldt-Universität zu Berlin \\
    \texttt{\{berkholz,keppelej,schweika\}@informatik.hu-berlin.de}
  }

}

\usepackage{amsthm}
\usepackage{amsmath} 
\usepackage{amsfonts} 
\usepackage{amssymb}   
\usepackage{stmaryrd} %
\usepackage{enumerate}

\title{%
Answering Conjunctive Queries under Updates%
\thanks{This is the full version of the conference contribution \cite{BKS_enumeration_PODS17}.}%
}%

       \newtheorem{theorem}{Theorem}[section] 
        \newtheorem{lemma}[theorem]{Lemma}

        \theoremstyle{definition}
        \newtheorem{definition}[theorem]{Definition}
        \newtheorem{claim}[theorem]{Claim}

        \newtheorem{example}[theorem]{Example}
        \newtheorem{conjecture}[theorem]{Conjecture}

\newcommand{\nc}[1]{\newcommand{#1}}
\newcommand{\rnc}[1]{\renewcommand{#1}}

\rnc{\leq}{\ensuremath{\leqslant}}
\rnc{\geq}{\ensuremath{\geqslant}}

\rnc{\le}{\leq}
\rnc{\ge}{\geq}

\nc{\isdef}{\ensuremath{:=}}
\nc{\deff}{\isdef}
\nc{\defi}{\isdef}

\nc{\set}[1]{\ensuremath{\{#1\}}}
\nc{\setsize}[1]{\ensuremath{|#1|}}
\nc{\Setsize}[1]{\ensuremath{\big|#1\big|}}
\nc{\Set}[1]{\ensuremath{\big\{#1\big\}}}
\nc{\setc}[2]{\set{#1 \ : \ #2}}
\nc{\Setc}[2]{\Set{#1 \ : \ #2}}

\nc{\aufgerundet}[1]{\ensuremath{\lceil #1 \rceil}}
\nc{\abgerundet}[1]{\ensuremath{\lfloor #1 \rfloor}}

\nc{\dcup}{\ensuremath{\dot\cup}}

\nc{\ov}[1]{\ensuremath{\overline{#1}}}

\nc{\NN}{\ensuremath{\mathbb{N}}}
\nc{\NNpos}{\ensuremath{\NN_{\scriptscriptstyle\geq 1}}}
\nc{\RR}{\ensuremath{\mathbb{R}}}
\nc{\RRpos}{\ensuremath{\RR_{\scriptscriptstyle\geq 0}}}

\nc{\und}{\ensuremath{\wedge}}
\nc{\Und}{\ensuremath{\bigwedge}}
\nc{\oder}{\ensuremath{\vee}}
\nc{\Oder}{\ensuremath{\bigvee}}
\nc{\nicht}{\ensuremath{\neg}}
\nc{\impl}{\ensuremath{\to}}
\nc{\gdw}{\ensuremath{\leftrightarrow}}

\newcommand{\eexists}{\,\exists}
\newcommand{\uund}{\,\und\,}
\newcommand{\body}[1]{\,\big(#1\big)} %
\newcommand{\bodyjoin}[1]{\big(#1\big)} 
\newcommand{\bbody}[1]{\;\big(\,#1\,\big)} %
\newcommand{\bbodyjoin}[1]{\big(\,#1\,\big)}

\nc{\free}{\ensuremath{\textrm{\upshape free}}}
\nc{\ar}{\ensuremath{\operatorname{ar}}}

\nc{\Structure}[1]{\ensuremath{\mathcal{#1}}}
\nc{\A}{\Structure{A}}
\nc{\B}{\Structure{B}}
\nc{\C}{\Structure{C}}

\nc{\isom}{\ensuremath{\cong}}

\nc{\querycont}{\ensuremath{\sqsubseteq}}
\nc{\eval}[2]{\ensuremath{#1(#2)}}
\nc{\semantik}[1]{\ensuremath{\left\llbracket#1\right\rrbracket}}

\nc{\CanDB}[1]{\ensuremath{\A_{#1}}} %
\nc{\CanTup}[1]{\ensuremath{t_{#1}}} %

\newcommand{\inds}{s}
\newcommand{\queryphi}{\varphi}
\newcommand{\varv}{v}
\newcommand{\varu}{u}
\newcommand{\varx}{x}
\newcommand{\vary}{y}
\newcommand{\varz}{z}
\newcommand{\varw}{w}

\newcommand{\sgpsi}{\psi} %
 
\newcommand{\relS}{S} %
 
\newcommand{\relT}{T} %

\newcommand{\relE}{E} %

\newcommand{\relR}{R} %
\newcommand{\smalleps}{\varepsilon}
\newcommand{\arityr}{r}

\newcommand{\actdomsize}{n}%

\newcommand{\setV}{V} %
\newcommand{\verta}{a} %
\newcommand{\vertb}{b} %
\newcommand{\vertc}{c} %
\nc{\Vars}{\ensuremath{\textrm{\upshape vars}}}
\nc{\atoms}{\ensuremath{\textrm{\upshape atoms}}}
\nc{\Adom}{\ensuremath{\textrm{\upshape adom}}}
\nc{\adom}[1]{\ensuremath{\Adom(#1)}} %
\nc{\dom}[1]{\ensuremath{\textrm{\upshape dom}(#1)}} %

\newcommand{\poly}{\operatorname{\textit{poly}}}

\newcommand{\qhier}{q-hie\-rar\-chi\-cal\xspace}
\newcommand{\Querytree}{\text{q-tree}\xspace}
\newcommand{\Querytrees}{\text{q-trees}\xspace}
\newcommand{\good}{good\xspace}

\newcommand{\OMv}{OMv\xspace}
\newcommand{\OMvcon}{\OMv{}-conjecture\xspace}
\newcommand{\OuMv}{OuMv\xspace}
\newcommand{\OV}{OV\xspace}
\newcommand{\OVcon}{\OV-conjecture\xspace}
\newcommand{\SETH}{SETH\xspace}
\newcommand{\indi}{i}
\newcommand{\indj}{j}
\newcommand{\indt}{t} %

\newcommand{\vecu}{\vec{u}}

\newcommand{\vecv}{\vec{v}}

\newcommand{\vecuset}{U}
\newcommand{\vecvset}{V}

\newcommand{\matM}{M}
\newcommand{\dimn}{n}

\newcommand{\presentcount}[2]{C^{#1}_{#2}}
\newcommand{\listcount}[2]{C^{#1}_{#2}}
\newcommand{\fitcount}[1]{C^{#1}}
\newcommand{\startcount}{C_{\text{\upshape start}}}
\newcommand{\extensionsetof}[1]{\mathcal E^{#1}}

\newcommand{\xtuple}{\overline{\varx}}
\newcommand{\listcountfree}[2]{\widetilde{C}^{#1}_{#2}}
\newcommand{\fitcountfree}[1]{\widetilde{C}^{#1}}
\newcommand{\startcountfree}{\widetilde{C}_{\text{\upshape start}}}
\newcommand{\extensionsetfreeof}[1]{\widetilde{\mathcal E}^{#1}}

\newcommand{\DBzero}[1]{#1}
\newcommand{\DBone}[1]{}

\newcommand{\bigoh}{O}
\newcommand{\bigOh}{\bigoh}
\newcommand{\trans}{^{\,\mkern-1.5mu\mathsf{T}}}

\newcommand{\parent}{\pointerfont{parent}}

\newcommand{\pa}{\ensuremath{\textrm{\upshape path}}}
\newcommand{\start}{\textit{start}}

\nc{\arrayfont}[1]{\ensuremath{\texttt{#1}}}
\newcommand{\larray}[1]{\ensuremath{\arrayfont{A}_{#1}}}

\newcommand{\ite}[3]{\ensuremath{\texttt{\upshape[}#1,#2,#3\texttt{\upshape]}}}

\newcommand{\query}{\ensuremath{\varphi}}
\newcommand{\qatom}{\ensuremath{\psi}}

\newcommand{\size}[1]{\ensuremath{|\!|#1|\!|}}
\nc{\card}[1]{\ensuremath{|#1|}}

\newcommand{\assign}{\ensuremath{\alpha}}
\newcommand{\assignb}{\ensuremath{\beta}}
\newcommand{\satisfy}[3]{\ensuremath{(#1,#2) \models #3}}
\newcommand{\notsatisfy}[3]{\ensuremath{(#1,#2) \not\models #3}}
\newcommand{\atm}[1]{\ensuremath{\textrm{\upshape rep}(#1)}}

\newcommand{\tree}{T}
\newcommand{\querytree}{T_{\queryphi}}
\newcommand{\hypergraph}{\mathcal H}

\newcommand{\hypedg}{e}%
\newcommand{\edges}{\operatorname{E}}
\newcommand{\queryhyper}{\hypergraph_{\queryphi}}
\newcommand{\claimast}{\textbf{(\textasteriskcentered)}\xspace}

\newcommand{\setX}{X}
\newcommand{\potenzmengeof}[1]{2^{#1}}

\newcommand{\sgpsix}{\sgpsi^\varx}
\newcommand{\sgpsiy}{\sgpsi^\vary}
\newcommand{\sgpsixy}{\sgpsi^{\varx,\vary}}
\newcommand{\psix}{\sgpsix}
\newcommand{\psiy}{\sgpsiy}
\newcommand{\psixy}{\sgpsixy}

\newcommand{\partPfull}{\mathcal P_{\queryphi,\dimn}}
\newcommand{\iotasubij}{\iota_{\indi,\indj}} 

\newcommand{\homh}{h}
\newcommand{\homg}{g}
\newcommand{\homDBtoquery}{g_{\queryphi,\dimn}}

\newcommand{\queryphicore}{\queryphi_{\text{\upshape core}}}

\newcommand{\arityk}{k}

\newcommand{\perm}{{\pi}}
\newcommand{\setcalR}{{\mathcal R}}

\newcommand{\setI}{I}

\nc{\insertp}{\textsc{Insert}}
\nc{\cleanup}{\textsc{cleanUp}}
\nc{\cleanups}{\textsc{cleanUp'}}

\newenvironment{mi}{\begin{enumerate}[$\bullet$]}{\end{enumerate}}

\nc{\myparagraph}[1]{\noindent\textbf{#1}. }
{\renewcommand{\myparagraph}[1]{\medskip\noindent\textbf{#1}. }}

\nc{\Yes}{\texttt{yes}}
\nc{\No}{\texttt{no}}

\nc{\Dom}{\ensuremath{\textbf{dom}}}
\nc{\Var}{\ensuremath{\textbf{var}}}
\nc{\schema}{\ensuremath{\sigma}}
\nc{\DB}{\ensuremath{D}} %
\nc{\DBstrich}{\ensuremath{D'}} %
\nc{\DBstart}{\ensuremath{{\DB_0}}} %
\nc{\DBempty}{\ensuremath{{\DB_{\emptyset}}}} %

\nc{\DS}{\ensuremath{\mathtt{D}}} %

\rnc{\phi}{\queryphi}

\nc{\UpdateFont}[1]{\ensuremath{\textsf{#1}}}
\nc{\Delete}{\UpdateFont{delete}}
\nc{\Insert}{\UpdateFont{insert}}
\nc{\Update}{\UpdateFont{update}}

\nc{\AlgoFont}[1]{\ensuremath{\textbf{#1}}}
\nc{\PREPROCESS}{\AlgoFont{preprocess}}
\nc{\INIT}{\AlgoFont{init}}
\nc{\UPDATE}{\AlgoFont{update}}
\nc{\ENUMERATE}{\AlgoFont{enumerate}}
\nc{\COUNT}{\AlgoFont{count}}
\nc{\ANSWER}{\AlgoFont{answer}}

\nc{\EOE}{\texttt{EOE}\xspace} %

\nc{\preprocessingtime}{\ensuremath{t_p}}
\nc{\inittime}{\ensuremath{t_i}}
\nc{\delaytime}{\ensuremath{t_d}}
\nc{\updatetime}{\ensuremath{t_u}}
\nc{\answertime}{\ensuremath{t_a}}
\nc{\countingtime}{\ensuremath{t_c}}

\nc{\phiBTypical}{\ensuremath{\phi'_{\relS\text{-}\relE\text{-}\relT}}}
\nc{\phiJTypical}{\ensuremath{\phi_{\relS\text{-}\relE\text{-}\relT}}}
\nc{\phiET}{\ensuremath{\phi_{\relE\text{-}\relT}}}

\nc{\restrict}[2]{\ensuremath{{#1}_{|#2}}}
\nc{\extend}[3]{\ensuremath{{#1}\tfrac{#3}{#2}}}
\nc{\valuation}{\ensuremath{\beta}}
\nc{\emptyassign}{\ensuremath{\emptyset}}
\nc{\Assign}[2]{\ensuremath{\frac{#2}{#1}}}

\nc{\vroot}{\ensuremath{\varv_{\textsl{root}}}}

\nc{\pointerfont}[1]{\textit{#1}}

\nc{\varitem}[1]{\ensuremath{v^{#1}}}
\nc{\assitem}[1]{\ensuremath{\assign^{#1}}}
\nc{\constitem}[1]{\ensuremath{a^{#1}}}
\nc{\parentitem}[1]{\ensuremath{\parent^{#1}}}
\nc{\childitem}[2]{\ensuremath{\pointerfont{child}^{#1}_{#2}}}
\nc{\llist}[2]{\ensuremath{\mathcal{L}_{#2}^{#1}}}
\nc{\startlist}{\ensuremath{\mathcal{L}_{\text{\upshape start}}}\xspace}
\nc{\nextlistitem}[1]{\ensuremath{\pointerfont{next-listitem}^{#1}}}
\nc{\prevlistitem}[1]{\ensuremath{\pointerfont{prev-listitem}^{#1}}}
\nc{\countitem}[1]{\ensuremath{C_{\textit{below}}^{#1}}}
\nc{\desc}[1]{\ensuremath{\text{desc}}}

\nc{\Null}{\ensuremath{0}}

\nc{\arrayA}{\arrayfont{A}}
\nc{\arrayB}{\arrayfont{B}}
\nc{\arrayC}{\arrayfont{C}}
\nc{\arrayE}{\arrayfont{E}}

\nc{\ITEMS}{\mathcal{I}}

\nc{\NIL}{\textsc{nil}}

\newcounter{CommentCounterRed}
\newcounter{CommentCounterBlue}
\newcounter{CommentCounterGreen}
\newcounter{CommentCounterGray}

\begin{document}

\maketitle{}

\makeatletter{}%
\begin{abstract}
We consider the task of enumerating and counting answers to $k$-ary
conjunctive queries against relational databases that may be updated
by inserting or deleting tuples.

We exhibit a new notion of 
\emph{q-hierarchical} conjunctive
\allowbreak
queries and show that these can be maintained efficiently in the following sense.
During a linear time preprocessing phase, we can build a data
structure that enables constant delay enumeration 
of the query results; and when the database is updated, we can update
the data structure and restart the enumeration phase within constant time.
For the special case of self-join free conjunctive queries we obtain
a dichotomy: if a query is not q-hierarchical, then query enumeration
with sublinear$^\ast$ delay and sublinear update time (and arbitrary
preprocessing time) is impossible. 

For answering Boolean conjunctive queries and 
for the more general problem of counting the number of
solutions of $k$-ary queries we obtain complete dichotomies:
if the query's homomorphic core is q-hierarchical,
then size of the  
the query result can be computed in linear time and maintained with
constant update time. Otherwise, the size of the query result cannot
be maintained with sublinear update time.

All our lower bounds rely on the OMv-conjecture, a conjecture on the
hardness of online matrix-vector multiplication that has recently
emerged in the field of fine-grained complexity to characterise the
hardness of dynamic problems.
The lower bound for the counting problem additionally relies on the
orthogonal vectors conjecture, which in turn is implied
 by the
strong exponential time hypothesis.

$^\ast)$ By \emph{sublinear} we mean $\bigoh(n^{1-\smalleps})$ for some
$\smalleps>0$, where $n$ is the size of the active domain of the current
database.
\end{abstract}

\ifthenelse{\isundefined{\ARTICLEFORMAT}}{
\keywords{%
 query evaluation, 
 constant delay enumeration, 
 counting complexity, 
 dichotomy,
 dynamic algorithms,
 online matrix-vector multiplication
}%
}{}

\makeatletter{}%

\newcommand{\DBsize}{m}
\newcommand{\Qsize}{\size{\query}}
\newcommand{\funcf}{f}
\newcommand{\FPT}{\mbox{\upshape FPT}\xspace}
\newcommand{\Wone}{\mbox{\upshape W[1]}\xspace}
\newcommand{\sharpWone}{\mbox{\upshape \#W[1]}\xspace}
\newcommand{\questeq}{\stackrel{?}{=}}
\newcommand{\questmodels}{\stackrel{?}{\models}}

\section{Introduction}

We study the algorithmic problem of answering a conjunctive query $\query$
against a dynamically changing relational database $\DB$.
Depending on the problem setting, we want to
answer a Boolean query,
count the number of output tuples of a non-Boolean query, or
enumerate the query result with constant delay.
We consider finite relational databases over a possibly infinite domain
in the fully dynamic setting where new tuples can be inserted or
deleted.

At the beginning, a dynamic query evaluation algorithm gets a query
$\query$ together with an initial database $\DBstart$.
It starts with a
preprocessing phase 
where a suitable data structure is built to represent
the result of evaluating $\query$ against $\DBstart$.
Afterwards, when the database is updated by inserting or deleting a
tuple, the data structure is updated, too, and the result of evaluating $\queryphi$ on the
updated database is reported.

The \emph{update time} is the time needed
to compute the representation of the new query result. In order to
be efficient, we require
that the update time is way smaller than the time needed to recompute the
entire query result.
In particular, we consider \emph{constant} update time that 
only depends on the query but not on the database, as feasible.
One can even argue that update time that scales
polylogarithmically ($\log^{\bigoh(1)}\size{\DB}$) with the size
$\size{\DB}$ of the database is feasible.
On the other hand, we regard
update time that scales polynomially ($\size{\DB}^{\Omega(1)}$) with
the database as infeasible.

This paper's aim is to classify those conjunctive queries (CQs, for short) that can be
efficiently maintained under updates, and to distinguish
them from queries that are hard in this sense.

\subsection{Our Contribution} 
\label{sec:contribution}

We identify a subclass of conjunctive queries that can be efficiently
maintained by a 
dynamic evaluation algorithm.  We call these queries \emph{\qhier}, and
 this notion is strongly related to the \emph{hierarchical}
property that was introduced by Dalvi and Suciu in \cite{Dalvi.2007}
and has already played a central role for efficient query evaluation
in various contexts (see Section~\ref{sec:mainresults} for a
definition and discussion of related concepts).  We show that after a linear time
preprocessing phase the result of any \qhier conjunctive query can be
maintained with constant update time.
This means that after every update we can answer a Boolean \qhier
query and compute the number of result tuples of a non-Boolean query
in constant time.
Moreover, we can  enumerate the query result with
constant delay.

We are also able to prove matching
lower bounds. These bounds are conditioned on the \emph{\OMvcon}, a conjecture on the hardness of online
matrix-vector multiplication that 
was introduced
by Henzinger, Krinninger, Nanongkai, and Saranurak in
\cite{Henzinger.2015} to
characterise the hardness of many dynamic problems.
The lower bound for the counting problem additionally relies on the
\emph{\OVcon}, a conjecture on the hardness of the orthogonal vectors
problem which in turn is implied by the well-known
strong exponential time hypothesis \cite{Williams.2005}.
We obtain the following dichotomies, which are
stated from the perspective of data complexity (i.e., the query is
regarded to be fixed) and hold for any fixed $\smalleps>0$.
By $n$ we always denote the size of the active domain of the current
database $\DB$.
For the enumeration problem we restrict our attention to
\emph{self-join free} CQs, where every relation symbol occurs
only once in the query.

\begin{theorem}\label{cor:dichotomyEnum}
  Let $\queryphi$ be a self-join free CQ.
\\
If $\queryphi$ is \qhier, then after a linear time preprocessing phase
the query result $\query(\DB)$ can be enumerated with
 constant delay and
 constant update time.
\\
 Otherwise, unless the \OMvcon fails,
 there is no dynamic
 algorithm that enumerates $\queryphi$ with arbitrary preprocessing
 time, and $\bigOh(n^{1-\smalleps})$ delay and update time.
\end{theorem}

\begin{theorem}\label{cor:dichotomyBool}
  Let $\queryphi$ be a Boolean CQ.
 \\
 If the homomorphic core of $\queryphi$ is \qhier, then 
 the query can
 be answered with
 linear
 preprocessing
 time and
 constant update time.
\\
 Otherwise, unless the \OMvcon fails,
 there is no
 algorithm that answers $\queryphi$ with arbitrary preprocessing time
 and $\bigOh(n^{1-\smalleps})$ update time.
\end{theorem}

\begin{theorem}\label{cor:dichotomyCount}
  Let $\queryphi$ be a CQ.
 \\
 If $\queryphi$ is \qhier, then the number
 $\setsize{\eval{\query}{\DB}}$ of tuples in the query result can be computed with  
 linear preprocessing time and constant update time.
 \\
 Otherwise, assuming the \OMvcon and the \OVcon, there is no algorithm that
 computes $\setsize{\eval{\query}{\DB}}$ with arbitrary preprocessing
 time and $\bigOh(n^{1-\smalleps})$ update time.
\end{theorem}

For the databases we construct in our lower bound proofs it holds that
$n\approx \sqrt{\size{\DB}}$.
Therefore all our lower bounds of the form $n^{1-\smalleps}$ translate
to $\size{\DB}^{\frac12-\smalleps}$ in terms of the size of the database.

\subsection{Related Work} 
\label{sec:related_work}

In more practically motivated papers
the task of answering a fixed query against a dynamic database has
been studied under the name
\emph{incremental view maintenance} (see e.\,g. \cite{Gupta.1993}).
Given the huge amount
of theoretical results on the complexity of query
evaluation in the static setting, surprisingly little is known about
the computational complexity of query evaluation under updates.

The \emph{dynamic descriptive complexity} framework introduced by
Patnaik and Immerman \cite{Patnaik.1997}
focuses on the expressive power of (fragments or extensions of)
first-order logic on dynamic databases and has led to a rich body of
literature (see \cite{Schwentick.2016} for a survey).
This approach, however, is quite different from the algorithmic
setting considered in the present paper,
as in every update step the internal data structure is updated by
evaluating a first-order query.
As this may take polynomial time even for the special case of conjunctive queries 
(as considered e.\,g.\ by Zeume and Schwentick in \cite{Zeume.2014}), this
is too expensive in the area of dynamic algorithms.

We are aware of only a few papers dealing with the computational
complexity of query evaluation under updates, namely
the study of XPath evaluation of Bj{\"{o}}rklund, Gelade, and Martens \cite{Bjorklund.2010} and
the studies of MSO queries on trees by
Balmin, Papakonstantinou, and Vianu \cite{DBLP:journals/tods/BalminPV04}
and by Losemann and Martens \cite{Losemann.2014}, the latter of which
is to the best of our knowledge the only work that deals with
efficient query enumeration under updates.%
\footnote{Let us mention that in a follow-up of the present paper, we
characterise the dynamic complexity of counting and
enumerating the results of first-order queries (and their extensions
by modulo-counting quantifiers) on bounded degree databases \cite{BKS-ICDT17}.}

In the static setting, a lot of research has been devoted to classify
those conjunctive queries that can be answered efficiently.
Below we give an overview of known results.

\myparagraph{Complexity of Boolean Queries}
The complexity of answering Boolean conjunctive
queries on a static database is fairly well understood.
For every fixed database schema $\schema$,
extending a result of \cite{DBLP:conf/stoc/GroheSS01}, Grohe
\cite{Grohe.2007} gave a tight 
characterisation of the \emph{tractable} CQs under the 
complexity theoretic assumption \FPT $\neq$ \Wone:
If we are given a Boolean CQ
$\query$ of size $\Qsize$ and a $\schema$-database $\DB$ of
size $\DBsize$, then $\query$ can
be answered against $\DB$ in time $\funcf(\Qsize)\cdot\DBsize^{\bigoh(1)}$ for
some computable function $\funcf$ if, and only if, the
homomorphic core of $\queryphi$ has bounded treewidth.
Marx \cite{Marx.2013} extended this classification to the case where
the schema is part of the input.

\myparagraph{Counting Complexity}
For computing the number of output tuples of a given join query (i.e.,
a quantifier-free CQ) 
over a fixed schema $\schema$, a characterisation
was proven by Dalmau and Jonsson \cite{Dalmau.2004}: Assuming \FPT $\neq$
\sharpWone, the output size $\Setsize{\eval{\queryphi}{\DB}}$ of a join query
$\queryphi$ evaluated on a $\schema$-database $\DB$ of
size $\DBsize$ can be
computed in time $\funcf(\Qsize)\cdot\DBsize^{\bigoh(1)}$ if, and only if,
$\queryphi$ has bounded treewidth.
The result has recently been extended to
all conjunctive queries over a fixed schema by Chen and Mengel \cite{Chen.2015}.
Structural properties that make the counting
problem for CQs tractable in the case where the schema
is part of the input have been identified in
\cite{Durand.2015,Greco.2014}.

\myparagraph{Join Evaluation}
When the entire result of a non-Boolean query has to be computed, 
the evaluation problem cannot be modelled as a decision or
counting problem and one has to come up with different measures to
characterise the hardness of query evaluation.
One approach that has been fruitfully applied to join evaluation
is to study the worst-case output size as a measure of
the hardness of a query.
Atserias, Grohe, and Marx\ \cite{Atserias.2013} identified the fractional edge
cover number of the join query
as a crucial measure for lower bounding its worst-case output size.
This bound was also shown to be optimal and is matched by so called ``worst-case optimal'' join
evaluation algorithms, see \cite{Ngo.2012,DBLP:conf/icdt/Veldhuizen14,DBLP:conf/pods/NgoNRR14,DBLP:conf/pods/KhamisNRR15}.

\myparagraph{Query Enumeration}
Another way of studying non-Boolean queries that is independent of the
actual or worst-case output size is \emph{query
enumeration}.
A query enumeration algorithm evaluates a non-Boolean query by
reporting, one by one without repetition, the tuples in the query
result.
The crucial measure to characterise queries that are tractable
w.r.t.\ enumeration is the delay between two output tuples.
In the context of constraint satisfaction, the combined complexity,
where
the query as well as the database are given as input, has been
considered. As the size
of the query
result might be exponential in the input size in this setting, queries that can be enumerated with \emph{polynomial
  delay} and polynomial preprocessing are regarded as ``tractable.''
Classes of conjunctive queries that can be enumerated with polynomial
delay have been identified in \cite{Bulatov.2012,Greco.2013}.
However, a complete
characterisation of conjunctive queries that are tractable in this
sense is not in sight.
 
More relevant to the database setting, where one 
evaluates a small query against a large database, is the notion of
\emph{constant delay enumeration} introduced by Durand and Grandjean
in \cite{DBLP:journals/tocl/DurandG07}. 
The preprocessing time is
supposed to be much smaller than
the time needed to evaluate the query
(usually,
linear in the size of the database), and the delay between two output
tuples may depend on the query, but not on the database.
A lot of research has been devoted to this subject, where one usually
tries to understand which structural restrictions on the query or on
the database allow constant delay enumeration.
For an introduction to this topic and an overview of the
state-of-the-art we refer the reader to the surveys
\cite{Segoufin.2015,Segoufin.2014,DBLP:conf/icdt/Segoufin13}. 

Bagan, Durand, and Grandjean \cite{Bagan.2007} showed that acyclic conjunctive
queries that are \emph{free-connex} can be enumerated with constant
delay after a linear time preprocessing phase (cf.\
\cite{DBLP:phd/hal/BraultBaron13} for a simplified proof of their
result).
They also showed that for self-join free acyclic conjunctive queries
the free-connex property is essential by proving the following lower bound.
Assume that multiplying two $n\times n$ matrices cannot
be done in time $\bigoh(n^2)$.
Then the result of a self-join free acyclic conjunctive query that is
not free-connex cannot be enumerated with constant delay after a
linear time preprocessing phase.

It turns out that our notion of \qhier conjunctive queries is a proper
subclass of the free-connex conjunctive queries.
Thus, there are queries that can be efficiently enumerated in the
static setting but are hard to maintain under database updates.

\myparagraph{Organisation} 
The remainder of the paper is structured as follows.
In Section~\ref{sec:preliminaries} we 
fix the basic notation along with the concept of dynamic algorithms
for query evaluation. Section~\ref{sec:mainresults} introduces
\emph{\qhier} queries and formally 
states our main theorems. 
We then present an alternative characterisation of \qhier queries in Section~\ref{sec:querytree} and prove our lower and upper bound theorems in
Sections~\ref{sec:lowerbounds} and \ref{sec:upperbound},
respectively. 
We conclude in Section~\ref{sec:discussion}.

\myparagraph{Acknowledgement}
We acknowledge the financial support by the German Research Foundation DFG under grant
SCHW~837/5-1.
The first author wants to thank Thatchaphol Saranurak and Ryan
Williams for helpful discussions on algorithmic conjectures.

\makeatletter{}%
\section{Preliminaries}\label{sec:preliminaries}

We write $\NN$ for the set of non-negative integers and let 
$\NNpos\deff\NN\setminus\set{0}$ and $[n]\deff\set{1,\ldots,n}$ for
all $n\in\NNpos$.
By $\potenzmengeof{M}$ we denote the power set of a set $M$.

\myparagraph{Databases}
We fix a countably infinite set $\Dom$, the \emph{domain} of potential
database entries. Elements in $\Dom$ are called \emph{constants}.
A \emph{schema} is a finite set $\schema$ of relation symbols, where
each $R\in\schema$ is equipped with a fixed \emph{arity} $\ar(R)\in\NNpos$. 
Let us fix a schema $\schema=\set{R_1,\ldots,R_s}$, and let
$r_i\deff\ar(R_i)$ for $i\in[s]$.
A \emph{database} $\DB$ of schema $\schema$ ($\schema$-db, for short), 
is of the form $\DB=(R_1^\DB,\ldots,R_s^\DB)$, where $R_i^\DB$ is a finite subset of 
$\Dom^{r_i}$.
The \emph{active domain} $\adom{\DB}$ of $\DB$ is the smallest subset
$A$ of $\Dom$ such that $R_i^\DB\subseteq A^{r_i}$ for all $i\in [s]$.

\myparagraph{Updates}
We allow to update a given database of schema $\schema$ by inserting or deleting
tuples as follows. An \emph{insertion} command is of the form
 \Insert\ $R(a_1,\ldots,a_r)$
for $R\in\schema$, $r=\ar(R)$, and $a_1,\ldots,a_r\in \Dom$. When
applied to a $\schema$-db $\DB$, it results in the updated $\schema$-db
$\DB'$ with $R^{\DB'}\deff R^{\DB}\cup\set{(a_1,\ldots,a_r)}$ and
$S^{\DB'}\deff S^{\DB}$ for all $S\in\schema\setminus\set{R}$.
A \emph{deletion} command is of the form
 \Delete\ $R(a_1,\ldots,a_r)$
for $R\in\schema$, $r=\ar(R)$, and $a_1,\ldots,a_r\in \Dom$. When
applied to a $\schema$-db $\DB$, it results in the updated $\schema$-db
$\DB'$ with $R^{\DB'}\deff R^{\DB}\setminus\set{(a_1,\ldots,a_r)}$ and
$S^{\DB'}\deff S^{\DB}$ for all $S\in\schema\setminus\set{R}$.
Note that both types of commands may change the database's active domain.

\myparagraph{Queries}
We fix a countably infinite set $\Var$ of \emph{variables}.
An \emph{atomic query} (for short: \emph{atom}) $\qatom$ of schema $\schema$ is of the form 
$R u_1 \cdots u_r$ with $R\in\schema$, $r=\ar(R)$, and
$u_1,\ldots,u_r\in\Var$.
The set of variables occurring in $\qatom$ is denoted by
$\Vars(\qatom)\deff \set{u_1,\ldots,u_r}$.
A \emph{conjunctive query} (CQ, for short) of schema $\schema$ is of the form 
\begin{equation}\label{eq:CQ}
  \exists y_1\,\cdots\,\exists y_\ell \bbody{
   \qatom_1 \uund \cdots \uund \qatom_d
  }
\end{equation}
where $\ell\in\NN$, $d\in\NNpos$, $\qatom_j$ is an atomic
query of schema $\schema$ for every $j\in [d]$, and
$y_1,\ldots,y_\ell$ are pairwise distinct elements in $\Var$.
\emph{Join queries} are quantifier-free CQs, i.e., CQs of the form
$\big(\qatom_1\und\cdots\und\qatom_d\big)$.
A CQ is called \emph{self-join free} (or \emph{non-repeating} or
\emph{simple}) if no relation symbol occurs more than once in the
query. 
For a CQ $\query$ of the form \eqref{eq:CQ} we let
$\Vars(\query)$ be the set of all variables occurring in
$\query$, and we let 
$\free(\query)\deff \Vars(\query)\setminus\set{y_1,\ldots,y_\ell}$ be
the set of \emph{free} variables. 

For $k\in\NN$, a \emph{$k$-ary conjunctive query} ($k$-ary CQ, for
short) is of the form
$\query(x_1,\ldots,x_k)$, where $\query$ is a CQ of schema $\schema$,
$k=|\free(\query)|$, and $x_1,\ldots,x_k$ is a list of the free
variables of $\query$. We will often assume that the tuple
$(x_1,\ldots,x_k)$ is clear from the context and simply write $\query$
instead of $\query(x_1,\ldots,x_k)$.
The semantics of CQs are defined as usual: \
A \emph{valuation} is a mapping $\valuation:\Var\to\Dom$.
For a $\schema$-db $\DB$ and an atomic query
$\qatom=Ru_1\cdots u_r$ 
we write $(\DB,\valuation)\models\qatom$ to indicate that 
$\big(\valuation(u_1),\ldots,\valuation(u_r)\big)\in
R^{\DB}$.
For a $k$-ary CQ $\query(x_1,\ldots,x_k)$ where $\query$ is of the
form \eqref{eq:CQ}, and for a tuple
$\ov{a}=(a_1,\ldots,a_k)\in\Dom^k$,
a valuation $\valuation$ is said to be \emph{compatible} with $\ov{a}$
iff $\valuation(x_i)=a_i$ for all $i\in [k]$.
We write $\DB\models\query[\ov{a}]$ to indicate that there is a valuation
$\valuation$ that is compatible with $\ov{a}$ such that $(\DB,\valuation)\models \qatom_j$
for all $j\in [d]$.
The \emph{query result} $\query(\DB)$ is
defined as the set of all tuples
$\ov{a}\in\Dom^k$ with $\DB\models\query[\ov{a}]$.
Clearly,
$\query(\DB)\subseteq\Adom(\DB)^k$.

A \emph{Boolean} CQ is a CQ $\query$ with
$\free(\query)=\emptyset$.
As usual, for Boolean CQs $\query$ we will write $\query(\DB)=\Yes$
instead of $\query(\DB)\neq\emptyset$, and $\query(\DB)=\No$ instead of 
$\query(\DB)=\emptyset$.

\myparagraph{Sizes and Cardinalities}
The \emph{size} $\size{\query}$ of a CQ $\query$ is defined as the
length of $\query$ when viewed as a word over the alphabet
$\schema\cup\Var\cup\set{\exists,\wedge,(,)}$. %
For a $k$-ary CQ $\query(x_1,\ldots,x_k)$ and a $\sigma$-DB $\DB$, the 
\emph{cardinality of the query result} is the number $|\query(\DB)|$ of
tuples in $\query(\DB)$.

The \emph{cardinality} $\card{\DB}$ of a $\schema$-db $\DB$ is defined
as the number of tuples stored in $\DB$, i.e.,
$\card{\DB}\deff\sum_{R\in\schema} |R^{\DB}|$. 
The \emph{size} $\size{\DB}$ of $\DB$ is defined as
$|\schema|+|\Adom(\DB)|+\sum_{R\in\schema} \ar(R){\cdot}
|R^D|$ and corresponds to the size of a reasonable encoding of $\DB$.
We will often write $n$
to denote the cardinality $|\Adom(\DB)|$ of
$\DB$'s active domain.

\myparagraph{Dynamic Algorithms for Query Evaluation}
We use Random Access Machines (RAMs) with $\bigoh(\log n)$ word-size and a uniform cost
measure as a model of computation.
In particular, adding and multiplying integers that are polynomial in
the input size can be done in constant time.
For our purposes it will be convenient to assume that $\Dom=\NNpos$.
We will assume that the RAM's memory is initialised to $\Null$. In
particular, if an algorithm uses an array, we will assume
that all array entries are initialised to $\Null$, and this initialisation
comes at no cost (in real-world computers this can be achieved by using the
\emph{lazy array initialisation technique}, cf.\ e.g.\ the textbook
\cite{MoretShapiro}). 
A further assumption 
that is unproblematic within the RAM-model, but
unrealistic for real-world computers, is that for every fixed
dimension $d\in\NNpos$ we have available an unbounded number of
$d$-ary arrays $\arrayA$ such that for given $(n_1,\ldots,n_d)\in\NN^d$
the entry $\arrayA[n_1,\ldots,n_d]$ 
at position $(n_1,\ldots,n_d)$ can be accessed in constant time.\footnote{While this can be accomplished easily in the RAM-model, 
for an implementation on real-world
computers one would probably have to resort to replacing our use of
arrays by using suitably designed hash functions.}

Our algorithms will take as input a $k$-ary CQ 
$\query(x_1,\ldots,x_k)$ and a $\schema$-db $\DBstart$.
For all query evaluation problems considered in this paper, we aim at
routines
$\PREPROCESS$ and $\UPDATE$ which achieve the following:
\begin{mi}
 \item upon input of $\query(x_1,\ldots,x_k)$ and $\DBstart$,
   $\PREPROCESS$ \allowbreak builds a data
   structure $\DS$ which represents $\DBstart$ (and which is designed in
   such a way that it supports efficient evaluation of $\query$ on $\DBstart$)
 \item upon input of a command
   $\Update\ R(a_1,\ldots,a_r)$ (with
   $\Update\in\set{\Insert,\Delete}$), 
   calling $\UPDATE$  modifies the data structure $\DS$ such that it
   represents the updated database $\DB$.
\end{mi}
The \emph{preprocessing time} $\preprocessingtime$ is the
time used for performing $\PREPROCESS$;
the \emph{update time} $\updatetime$ is the time used for performing
an $\UPDATE$.

In the following, $\DB$ will always denote the database that is
currently represented by the data structure $\DS$.
To solve the \emph{enumeration problem under updates}, apart from the
routines $\PREPROCESS$ and $\UPDATE$ we 
aim at a routine $\ENUMERATE$ such that
calling $\ENUMERATE$ invokes an enumeration of all tuples
(without repetition) that belong to
the query result $\query(\DB)$.
The \emph{delay} $\delaytime$ is the maximum time
used during a call of $\ENUMERATE$
\begin{mi}
\item until the output of the first tuple (or the end-of-enumeration
  message $\EOE$, if $\phi(D)=\emptyset$),
\item between the output of two consecutive tuples, and 
\item between the output of the last tuple and the end-of-enumeration
  message $\EOE$.
\end{mi}

To solve the \emph{counting problem under updates}, instead of
$\ENUMERATE$ we aim at a routine $\COUNT$ which outputs the cardinality
$|\query(\DB)|$ of the query result.
The \emph{counting time} $\countingtime$ is the time used for
performing a $\COUNT$.
To \emph{answer} a \emph{Boolean} conjunctive query under updates,
instead of $\ENUMERATE$ or $\COUNT$ we aim at a routine $\ANSWER$
that produces the answer $\Yes$ or $\No$ of $\query$ on $\DB$.
The \emph{answer time} $\answertime$ is the time used for
performing $\ANSWER$.
Whenever speaking of a \emph{dynamic algorithm}, we mean an algorithm
that has routines $\PREPROCESS$ and $\UPDATE$ and, depending on the
problem at hand, at least one of the routines $\ENUMERATE$, $\COUNT$,
and $\ANSWER$.

Throughout the paper, we often adopt the view of \emph{data
  complexity} and use the $\bigoh$-notation to suppress factors that
may depend on the query but not on the database. For example, 
``linear preprocessing time'' \allowbreak means 
$\preprocessingtime= f(\phi)\cdot\size{\DBstart}$ and 
``constant update time'' means  $\updatetime=f(\query)$, for a function
$f$ with codomain $\NN$. When writing $\poly(\queryphi)$ we mean $\size{\queryphi}^{\bigOh(1)}$.

\makeatletter{}%

\section{Main Results}
\label{sec:mainresults}

\newcommand{\queryphialt}{\queryphi'}
\newcommand{\queryalttree}{\tree_{\queryphialt}}

Our notion of \emph{\qhier} conjunctive queries is related to the
\emph{hierarchical} property that has already played a central role for
efficient query evaluation in various contexts. It has been introduced
by Dalvi and Suciu in \cite{Dalvi.2007} to characterise the 
Boolean CQs that can be answered in polynomial time on
probabilistic databases. 
They obtained a dichotomy stating for self-join free queries that the complexity of query
evaluation on probabilistic databases is in PTIME for hierarchical
queries and \#P-complete for non-hierarchical queries.
Fink and Olteanu \cite{Fink.2016} generalised the notion and the
dichotomy result to non-Boolean queries and to queries using negation. 
In the different context of query evaluation on massively parallel
architectures, Koutris and Suciu \cite{KoutrisSuciuPODS11} 
considered hierarchical join queries and singled out a subclass of
so-called tall-flat queries as exactly those queries that can be
computed with only one broadcast step in their \emph{Massively
  Parallel} model of query evaluation. For further information on the
various uses of the hierarchical property we refer the reader to \cite{Fink.2016}.

The definition of {hierarchical} queries relies on the following notion.
Consider a CQ $\queryphi$ of the form \eqref{eq:CQ}.
For every variable $\varx\in\Vars(\queryphi)$ we let $\atoms(\varx)$
be the set of all atoms $\sgpsi_j$ of $\queryphi$ such that  $\varx\in \Vars(\sgpsi_j)$.
Dalvi and Suciu \cite{Dalvi.2007} call a Boolean CQ $\queryphi$ \emph{hierarchical} iff
the condition 
\begin{itemize}
 \item[\ \ $(*)$: ]
  $\atoms(\varx)\subseteq\atoms(\vary)$ \ or \ 
  $\atoms(\varx)\supseteq\atoms(\vary)$ \ or \ 
  $\atoms(\varx)\cap\atoms(\vary)=\emptyset$
\end{itemize}
is satisfied by all variables $x,y\in\Vars(\queryphi)$.
An example for a hierarchical Boolean CQ is 
\;$\exists\varx\exists\vary\exists\varz\exists\vary'\exists\varz'\body{
\relR\varx\vary\varz \,\wedge\, \relR\varx\vary\varz' \,\wedge\,
 \relE\varx\vary \,\wedge\, \relE\varx\vary'}$.

In \cite{KoutrisSuciuPODS11}, Koutris and Suciu transferred the notion to join queries $\queryphi$, which they
call hierarchical iff condition $(*)$ is satisfied by all variables $x,y\in\Vars(\queryphi)$.
In \cite{Fink.2016}, Fink and Olteanu introduced a slightly different notion for a more general class of queries. Translated 
into the setting of CQs, their notion (only) requires that condition $(*)$ is satisfied by all \emph{quantified} variables, i.e., 
variables $x,y\in\Vars(\queryphi)\setminus\free(\queryphi)$.
Obviously, both notions coincide on \emph{Boolean} CQs, but on \emph{join queries} Koutris and Suciu's notion 
is more restrictive than Fink and Olteanu's notion (according to which \emph{all} quantifier-free CQs are hierarchical).
For example, the join query
\begin{equation}
    \label{eq:SxExyTy-join}
 \phiJTypical \quad \deff\quad 
 \bbodyjoin{\relS\varx \uund \relE\varx\vary  \uund \relT\vary}
\end{equation}
is hierarchical w.r.t.\ Fink and Olteanu's notion, and non-hierarchical w.r.t.\ Koutris and Suciu's notion.

In the context of answering queries under updates, our lower bound results show that
the join query $\phiJTypical$, as well as its Boolean version 
\begin{equation}
    \label{eq:SxExyTy}
 \phiBTypical \quad\deff\quad \exists\varx\eexists\vary
 \bbody{\relS\varx \uund \relE\varx\vary \uund \relT\vary}
\end{equation}
are intractable.
A further query that is hierarchical, but intractable in our setting is
\begin{equation}
  \label{eq:ExyTy}
\phiET \quad \deff \quad  \exists\vary \bbody{ \relE\varx\vary \uund
  \relT\vary}.
\end{equation}
To ensure tractability of a conjunctive query in our setting, we will require that its 
quantifier-free part is hierarchical in Koutris and Suciu's notion and, additionally,
the quantifiers respect the query's hierarchical form. We call such queries \emph{q-hierarchical}.

\begin{definition}\label{def:consquanthier}
  A CQ $\queryphi$ is \emph{\qhier} if for any two
  variables $\varx,\vary\in\Vars(\queryphi)$ the following is satisfied:
  \begin{enumerate}[(i)]
  \item\label{item:hier-cond}
  $\atoms(\varx)\subseteq\atoms(\vary)$ \ or \
  $\atoms(\varx)\supseteq\atoms(\vary)$ \ or \ 
  $\atoms(\varx)\cap\atoms(\vary)=\emptyset$\,,\ \  and
  \item\label{item:quant-cond}
  if $\atoms(\varx)\subsetneq\atoms(\vary)$ and 
  $\varx\in\free(\queryphi)$, then $\vary\in\free(\queryphi)$.
  \end{enumerate}
\end{definition}

Note that a \emph{Boolean} CQ is \qhier iff
it is hierarchical, and a join query is \qhier iff it is hierarchical w.r.t.\ Koutris and Suciu's notion.
The queries $\phiJTypical$ and $\phiET$ are minimal examples for queries that are not q-hierarchical because they do 
not satisfy condition~\eqref{item:hier-cond} and \eqref{item:quant-cond}, respectively.
Regarding the query $\phiET$, note
that all other versions such as the query
$\exists\varx \body{ \relE\varx\vary \, \und\, \relT\vary}$, 
the join query $\bodyjoin{ \relE\varx\vary \,\und\, \relT\vary}$,
and 
the Boolean query $\exists\varx\exists\vary \body{ \relE\varx\vary \,\und\,\relT\vary}$,
are \mbox{\qhier}.

It is not hard to see that we can decide in
polynomial time whether a given CQ is \qhier (see Lemma~\ref{lem:querytree}).
Our first main result 
shows that all \qhier CQs can be efficiently maintained under database updates:

\begin{theorem}\label{thm:upperbound}
  There is a dynamic algorithm that receives a \qhier conjunctive query $\queryphi$ and a $\schema$-db
  $\DBstart$, and computes within $\preprocessingtime =
  \poly(\queryphi){\cdot}\bigOh(\size{\DBstart})$ 
  preprocessing 
  time a data structure that can be updated in time
  $\updatetime = 
  \poly(\queryphi)$
  and allows to 
  \begin{enumerate}[(a)]
  \item enumerate $\eval{\queryphi}{\DB}$ with delay
    $\delaytime = 
     \poly(\queryphi)$, 
    \,and
  \item compute the cardinality $|\eval{\queryphi}{\DB}|$ in time
    $\countingtime = \bigoh(1)$,
  \end{enumerate}
  where $\DB$ is the current database.
\end{theorem}

Note that 
this implies that \qhier \emph{Boolean} conjunctive queries can be
answered in constant time.
Our algorithm crucially relies on the tree-like structure of
hierarchical queries, which has already been used for efficient query
evaluation in \cite{DBLP:journals/vldb/DalviS07,Dalvi.2007,KoutrisSuciuPODS11,Fink.2016,DBLP:journals/tods/OlteanuZ15}.
In Section~\ref{sec:querytree} we present the notion of a \emph{\Querytree}
and show that it precisely characterises the \emph{\qhier} conjunctive queries.
These \Querytrees serve as a basis for the data structure
used in our dynamic algorithm for query answering.
Details on this algorithm along with a proof of
Theorem~\ref{thm:upperbound} can be found in 
Section~\ref{sec:upperbound}.
Let us mention that 
every \Querytree is an \emph{f-tree} in the sense of
\cite{DBLP:journals/tods/OlteanuZ15}, but there exist f-trees that are
no \Querytrees.
The 
dynamic data structure that is computed by our algorithm can be viewed
as an
\emph{f-representation} of the query result
\cite{DBLP:journals/tods/OlteanuZ15}, but not every f-representation
can be efficiently maintained under database updates.

\medskip

We now discuss our further main results, which show that the 
\emph{\qhier} property is necessary
for designing efficient dynamic algorithms, and that the results from
Theorem~\ref{thm:upperbound} cannot be extended to queries that are
not \qhier.
As discussed in the introduction, our lower bounds rely on the \OMvcon
and the \OVcon.
For more details on these conjectures, as well as proofs of our lower
bound theorems,
we refer the reader to
Section~\ref{sec:lowerbounds}. 
In the following, $\DBstart$ denotes the initial database that serves
as input for the $\PREPROCESS$ routine, and
$\actdomsize=\setsize{\Adom(\DB)}$ denotes the size  of the active
domain of a dynamically changing database $\DB$.
Our first lower
bound theorem states that non-\qhier \emph{self-join free} conjunctive queries
cannot be enumerated efficiently under updates.

\begin{theorem}\label{thm:enumerating_CQ_intro}
  Fix a number $\smalleps>0$ and a self-join free conjunctive query $\queryphi$.
  If
  $\queryphi$ is not \qhier, then there is no
  algorithm with arbitrary preprocessing time and
  $\bigoh(\actdomsize^{1-\smalleps})$ update time that enumerates 
  $\eval{\queryphi}{\DB}$ with
  $\bigoh(\actdomsize^{1-\smalleps})$ delay, unless the \OMvcon fails.
\end{theorem}

For \emph{Boolean} CQs we obtain a lower bound for
\emph{all}  queries, i.e., 
also for queries that are not self-join free.
To state the result, we need the standard notion of a homomorphic core.
A \emph{homomorphism} from a conjunctive query 
$\queryphi(\varx_1,\ldots,\varx_\arityk)$ to a
conjunctive query $\queryphialt(\vary_1,\ldots,\vary_\arityk)$ is a
mapping 
$h$ from $\Vars(\queryphi)$ to $\Vars(\queryphialt)$ 
such that
$\homh(\varx_i) = \vary_i$ for all $i\in\set{1,\ldots,k}$, and if
$\relR \varu_1\cdots \varu_\arityr$ is an atom of $\queryphi$, then
$\relR\, \homh(\varu_1) \cdots \homh(\varu_\arityr)$ is an atom of
$\queryphialt$.
The \emph{homomorphic core} (for short, \emph{core}) of a conjunctive query $\queryphi$ is a minimal
subquery $\queryphialt$ of $\queryphi$ such that there is a homomorphism from
$\queryphi$ to $\queryphialt$, but no homomorphism from $\queryphialt$
to a proper subquery of $\queryphialt$.
By Chandra and Merlin's homomorphism theorem, every CQ $\phi$ has a unique (up to isomorphism)
core $\phi'$, and $\phi'(\DB)=\phi(D)$ for all databases $\DB$ (cf.,
e.\,g., \cite{AHV-Book}). 
While self-join free queries are their own cores, the situation is different
for general CQs.
Consider, for example, the queries
\begin{align*}
\queryphi \ &\deff \quad \exists\varx\eexists\vary \bbody{ \relE\varx\varx \uund
            \relE\varx\vary \uund \relE\vary\vary}\quad \text{ and } \\ 
\queryphialt \ &\deff \quad \exists\varx \bbody{ \relE\varx\varx }.
\end{align*}
Here, $\queryphialt$ is a core of $\queryphi$ and thus
$\eval{\queryphi}{\DB}=\eval{\queryphialt}{\DB}$ for every database $\DB$.
However, $\queryphialt$ is \qhier, whereas $\queryphi$ is not.
The next lower bound theorem states that the result of a Boolean
conjunctive query cannot
be maintained efficiently if the query's core is not \qhier.

\begin{theorem}\label{thm:modelchecking_joins_intro}
  Fix a number $\smalleps>0$ and a Boolean conjunctive query
  $\queryphi$.  If the homomorphic core of $\queryphi$ is not \qhier,
 then there is no
  algorithm with arbitrary preprocessing time and
  $\updatetime=\bigoh(\actdomsize^{1-\smalleps})$
  update time that
  answers 
  $\eval{\queryphi}{\DB}$ in time
$\answertime=\bigoh(\actdomsize^{2-\smalleps})$, unless the \OMvcon
  fails.
\end{theorem}

Let us now turn to the problem of 
computing the cardinality $\setsize{\eval{\queryphi}{\DB}}$ of the
result of a 
query $\queryphi(x_1,\ldots,x_k)$.
From the Theorems~\ref{thm:upperbound} and
\ref{thm:modelchecking_joins_intro} we know that 
we can efficiently decide whether $\setsize{\eval{\queryphi}{\DB}}>0$
if, and only if, the homomorphic core of
$\exists x_1 \cdots \exists x_\arityk\,\queryphi$ is \qhier.
The complexity of actually counting the number of tuples in
$\queryphi(\DB)$, however, depends on whether
the core of the query $\queryphi(\varx_1,\ldots,\varx_\arityk)$ itself (rather
than the core of its Boolean version $\exists x_1 \cdots \exists
x_\arityk\,\queryphi$) is \qhier. %
As in the Boolean case, the next theorem (together with Theorem~\ref{thm:upperbound}) implies a dichotomy for
all conjunctive queries. One difference is that we have to additionally rely on the \OVcon.

\begin{theorem}\label{thm:counting_CQ_lowerbound}
Fix a number $\smalleps>0$ and a conjunctive query $\queryphi$.
If the homomorphic core of $\queryphi$ is not \qhier, then there is no algorithm with
arbitrary
preprocessing time and
$\updatetime=\bigoh(\actdomsize^{1-\smalleps})$ update time that computes
$\setsize{\eval{\queryphi}{\DB}}$ 
in time $\countingtime=\bigoh(\actdomsize^{1-\smalleps})$,
assuming the \OMvcon and the \OVcon.
\end{theorem}

Combining Theorem~\ref{thm:upperbound} with the
Theorems~\ref{thm:enumerating_CQ_intro},
\ref{thm:modelchecking_joins_intro}, and
\ref{thm:counting_CQ_lowerbound} immediately leads to the 
dichotomies (Theorems~\ref{cor:dichotomyEnum},
\ref{cor:dichotomyBool}, and \ref{cor:dichotomyCount}) stated in
the introduction.

\makeatletter{}%

\ifthenelse{\isundefined{\ARTICLEFORMAT}}{
\section{The tree-like structure of  \\ \mbox{q}-hierarchical queries}
}{
\section{The tree-like structure of  \mbox{q}-hierarchical queries}
}

\label{sec:querytree}

We now give an alternative characterisation of \qhier queries that
sheds more light on their ``tree-like'' structure and will be useful for designing
efficient query evaluation algorithms.
We say that a CQ $\queryphi$ is \emph{connected} if for any two
variables $\varx, \vary \in \Vars(\queryphi)$ there is a path $\varx =
\varz_0, \ldots, \varz_\ell = \vary$ such that for each $j<\ell$ there
is an atom 
$\sgpsi$ of $\queryphi$ such that
$\set{\varz_j,\varz_{j+1}}\subseteq \Vars(\sgpsi)$.
Note that every conjunctive query can be written as a
conjunction $\bigwedge_i \queryphi_i$ of
connected conjunctive queries $\queryphi_i$ over pairwise disjoint variable sets.
We call these $\queryphi_i$ the \emph{connected components} of the query.
Note that a query is \qhier if, and only if, all its connected
components are \qhier.
Next, we define the notion of a \emph{\Querytree} for a connected
query $\phi$ and show that $\phi$ is \qhier iff it has a \Querytree.

\begin{figure}
  \centering
   \begin{tikzpicture}[scale=0.95]
     \node(x1) at (0,3) {$\varx_1$}; \node(x2) at (0,2) {$\varx_2$};
     \node(x3) at (-0.5,1) {$\varx_3$};
     \node[ %
     fill=black!20] (x4) at (0.5,1) {$\varx_4$};
     \node[%
     fill=black!20] (x5) at (-0.5,0) {$\varx_5$} ; \draw[->] (x1) --
     (x2); \draw[->] (x2) -- (x3); \draw[->] (x2) -- (x4); \draw[->]
     (x3) -- (x5);
     \begin{scope}[xshift=3cm]     
  \node(x1) at (0,3) {$x_2$};
  \node(x2) at (0,2) {$x_1$};
  \node(x3) at (-0.5,1) {$x_3$};
  \node[fill=black!20](x4) at (0.5,1) {$x_4$};
  \node[fill=black!20](x5) at (-0.5,0) {$x_5$};
  \draw[->] (x1) -- (x2);
  \draw[->] (x2) -- (x3);
  \draw[->] (x2) -- (x4);
  \draw[->] (x3) -- (x5);
    \end{scope}
 \end{tikzpicture}
 \caption{Two \Querytrees for
   $\varphi(\varx_1,\varx_2,\varx_3) = \exists \varx_4\exists \varx_5
   \body{
    \relE\varx_1\varx_2
    \und \relR\varx_4\varx_1\varx_2\varx_1
    \und \relR\varx_5\varx_3\varx_2\varx_1}$}
  \label{fig:querytree}
\end{figure}

\begin{definition}%
Let $\queryphi$ be a connected CQ.
A \emph{\Querytree} for $\queryphi$ 
is a rooted directed tree $\querytree=(\setV,\relE)$ with
$\setV=\Vars(\queryphi)$ where
\begin{enumerate}[(1)]
\item for all atoms $\sgpsi$ in
     $\queryphi$ the set $\Vars(\sgpsi)$ forms a directed path in $\querytree$ that
     starts from the root, and
\item if $\free(\queryphi)\neq\emptyset$, then $\free(\queryphi)$ is
     a connected subset in $\querytree$ that contains the root. 
\end{enumerate}
\end{definition}
\noindent
See Figure~\ref{fig:querytree} for examples of \Querytrees.
The following lemma gives a characterisation of the \emph{\qhier}
conjunctive queries via \Querytrees.

\begin{lemma}\label{lem:querytree}
  A CQ $\queryphi$ is \qhier if, and only if, every
  connected component of $\queryphi$
   has a \Querytree.
  Moreover, there is a polynomial time algorithm which decides whether
  an input CQ $\queryphi$ is \qhier, and if so, outputs a \Querytree for each
  connected component of $\queryphi$.
\end{lemma}
To prove the lemma we inductively apply the following claim.
\begin{claim}
  \label{claim:universal_variable_querystyle}
    For every
    connected \qhier CQ $\queryphi$ there is a variable $\varx\in\Vars(\queryphi)$
    that is contained in every atom of $\queryphi$.  Moreover, if
    $\free(\queryphi)\neq\emptyset$, then $\varx\in\free(\queryphi)$.
  \end{claim}

\begin{proof}%
  For simplicity, we associate with every conjunctive query $\queryphi$ the hypergraph
$\queryhyper$ with vertex set $\Vars(\queryphi)$
and hyperedges $\hypedg_{\sgpsi}\defi \Vars(\sgpsi)$ for every atom
$\sgpsi$ in $\queryphi$.
For a variable $\varx$ we let $\edges(\varx)\defi
\setc{\hypedg_{\sgpsi}}{\varx\in\Vars(\sgpsi)}$ be the set of
hyperedges that contain $\varx$.
  Let us recall some basic notation concerning hypergraphs.
  A \emph{path} of length $\ell$ in $\queryhyper$ is a sequence of
  variables $\varx_0,\ldots,\varx_\ell$ such that for every $i<\ell$
  there is a hyperedge containing $\varx_i$ and $\varx_{i+1}$. Two
  variables have \emph{distance} $\ell$ if they are connected by a path of
  length $\ell$, but not by a path of length $\ell-1$.
  We first show that \claimast every pair of hyperedges in $\queryhyper$ has a non-empty
  intersection.
  Suppose for contradiction that there are two hyperedges $\hypedg_1$
  and $\hypedg_2$ with $\hypedg_1\cap\hypedg_2=\emptyset$, let
  $\varx_0\in \hypedg_1$ and $\varx_\ell\in \hypedg_2$ be two variables of
  distance $\ell\geq 2$, and  $\varx_0,\ldots,\varx_\ell$ be a
  shortest path connecting both variables.
  Hence, for every 
  $i<\ell$ there is a hyperedge containing $\varx_i$ and $\varx_{i+1}$
  but no other variable from the path. 
  Therefore it holds that $\edges(x_0)\cap\edges(x_1)\neq \emptyset$.
  Furthermore, we have
  $\hypedg_1\notin\edges(\varx_1)$ and hence
  $\edges(\varx_0)\not\subseteq\edges(\varx_1)$.
  On the other hand, the hyperedge containing $\varx_1$ and $\varx_2$
  does not contain $\varx_0$ and therefore
  $\edges(\varx_1)\not\subseteq\edges(\varx_0)$, which contradicts the
  assumption that $\queryphi$ is \qhier.

  We now prove that there is a variable that is contained in every
  hyperedge (and hence in every atom of $\query$).
  We consider two cases.
  First suppose that for every pair of hyperedges $\hypedg_i$,
  $\hypedg_j$ it holds that either $\hypedg_i\subseteq \hypedg_j$ or
  $\hypedg_j\subseteq \hypedg_i$. Then there is a minimal hyperedge
  $\hypedg$ that is contained in every other hyperedge, and thus all of the
  variables in this hyperedge are contained in all other hyperedges as well.
  
  Now suppose that there are two hyperedges $\hypedg_i$,
  $\hypedg_j$ such that $\hypedg_i\not\subseteq \hypedg_j$ and
  $\hypedg_j\not\subseteq \hypedg_i$.
  By \claimast, both hyperedges have a non-empty intersection. Thus, we can
  choose some $\varx\in\hypedg_i\cap\hypedg_j$.
  We want to argue that $\varx$ is contained in every hyperedge of 
  $\queryhyper$
  and assume for contradiction that there is a hyperedge $\hypedg_k$ that
  does \emph{not} contain $\varx$.
  By \claimast we can choose some $\vary\neq\varx$ that is contained
  in the non-empty intersection of $\hypedg_j$ and $\hypedg_k$.
  But now we have $\hypedg_j \in \edges(\varx)\cap \edges(\vary)$,
  $\hypedg_i \in \edges(\varx)\setminus \edges(\vary)$, and $\hypedg_k \in
  \edges(\vary)\setminus \edges(\varx)$, contradicting that
  $\queryphi$
  is \qhier.

  Let $S$ be the set of all variables that are contained in every
  hyperedge.
  We have already shown that $S\neq\emptyset$.
  To ensure that there is a free variable in $S$, note that (by the definition of \emph{\qhier} CQs) if
  $\varx\in\free(\query)$ and $\varx\notin S$, then
  $S\subset \free(\query)$.
  Hence, $\free(\query)\neq\emptyset$ implies that we can choose a
  variable from  $\free(\query)\cap S$ that satisfies the conditions
  of the claim.
\end{proof}

\begin{proof}[Proof of Lemma~\ref{lem:querytree}]
The proof of the ``if'' direction is easy, as every connected
component that has a \Querytree $T$ must be \qhier, because
if $\vary$ is a descendant of $\varx$ in $T$, then 
$\atoms(\vary)\subseteq\atoms(\varx)$.  
For proving the ``only if'' direction of Lemma~\ref{lem:querytree} we
inductively apply Claim~\ref{claim:universal_variable_querystyle} to
construct a \Querytree $\querytree$ for all \emph{connected}
conjunctive queries
$\queryphi$ with at most $\ell$ variables.
The induction start for empty queries is trivial. For the induction step, assume
  that there is a \Querytree for every connected \qhier query with at
  most $\ell$ variables, 
  and let $\queryphi$ be a connected \qhier query  with
  $\ell+1$ variables.
  By Claim~\ref{claim:universal_variable_querystyle} there is at least one variable
  that
  is contained in every atom, and if $\free(\queryphi)\neq\emptyset$
  there is a free variable with this property.
  We choose such a variable $\varx$ (preferring free over quantified variables)
  and let $\varx$ be the root of $\querytree$.

  Now we consider the query $\queryphi'$ that is obtained from
  $\queryphi$ by ``removing'' $\varx$ from every atom.
  As this query is still \qhier,
  we can find by induction a tree $\tree_i$
  for every connected component $\queryphi'_i$ of $\queryphi'$.
  We let $\querytree$ be the disjoint union of the $\tree_i$ together
  with the root $\varx$ and conclude the construction  by adding an edge from
  $\varx$ to the root of each $\tree_i$.
  It is easy to see that this construction can be computed in
  polynomial time.
\end{proof}

\makeatletter{}%

\newcommand{\ellalt}{{\ell'}}
\newcommand{\phione}{\phiBTypical}
\newcommand{\phitwo}{\phiET}
\newcommand{\phiBool}{\queryphi_\exists}
\newcommand{\domn}{\ensuremath{\operatorname{dom}_\dimn}}

\section{Lower Bounds}
\label{sec:lowerbounds}

\subsection{The \OMvcon}
\label{sec:OMv}

We write $\vec{w}_i$ to denote the $i$-th component of an $n$-dimensional vector $\vec{w}$,
and we write $M_{i,j}$ for the entry in row $i$ and column $j$ of an
$n\times n$ matrix $M$.

We consider matrices and vectors over $\set{0,1}$.
All the arithmetic is done over the Boolean semiring, where multiplication
means conjunction and addition means disjunction.  For example,
for $n$-dimensional vectors $\vecu$ and $\vecv$ we have
$\vecu\trans\vecv=1$ if and only if there is an $\indi\in[n]$ such
the $\vecu_i=\vecv_i=1$.
Let $\matM$ be an
$\dimn\times\dimn$ matrix and let $\vecv^{\,1},\ldots, \vecv^{\,\dimn}$ be a
sequence of $n$ vectors, each of which has dimension $\dimn$.  The
\emph{online matrix-vector multiplication problem} is the following
algorithmic task. 
At first, the algorithm gets an $\dimn\times\dimn$ matrix $\matM$ and is
allowed to do some preprocessing.  Afterwards, the algorithm receives
the vectors $\vecv^{\,1},\ldots, \vecv^{\,\dimn}$ one by one and has to output
$\matM{\vecv^{\,\indt}}$ before it has access to $\vecv^{\,\indt+1}$
(for each $t<n$).  The
running time is the overall time the algorithm needs to produce the
output $\matM\vecv^{\,1},\ldots,\matM \vecv^{\,\dimn}$.

It is easy to see that this problem can be solved in
$\bigoh(\dimn^3)$ time; the best known algorithm runs in time
$\bigoh(\dimn^3/\log^2\dimn)$ \cite{Williams.2007}. 
The \emph{\OMvcon} was introduced by
Henzinger, Krinninger, Nanongkai, and Saranurak
in \cite{Henzinger.2015} and states that the online matrix-vector
multiplication problem cannot be solved in ``truly subcubic'' time
$\bigoh(n^{3-\smalleps})$ for any $\smalleps>0$.  Note that the
hardness of online matrix-vector multiplication crucially depends on
the requirement that the algorithm does not receive all vectors
$\vecv^{\,1},\ldots,\vecv^{\,n}$ at once.  In fact, without this requirement the output could
be computed in time $\bigoh(\dimn^{3-\smalleps})$ by using any fast
matrix multiplication algorithm.  The \OMvcon has been used to prove
conditional lower bounds for various dynamic problems and is a common
barrier for improving these algorithms, see \cite{Henzinger.2015}.
Contrary to classical complexity theoretic assumptions
such as $P\neq NP$, this conjecture shares with other recently
proposed algorithmic conjectures the less desirable fact that it can
hardly be called ``well-established''.  However, at least we know that
improving dynamic query evaluation algorithms for queries that are hard under
the \OMvcon is a very difficult task and (even if not completely
inconceivable) would lead to major breakthroughs in algorithms for
e.g.\ matrix multiplication (see \cite{Henzinger.2015} for a
discussion).

A variant of \OMv that is useful as an intermediate step in our
reductions is the following \OuMv problem.  Again, we are given an
$\dimn\times\dimn$ matrix $\matM$ and are allowed to do some
preprocessing.  Afterwards, a sequence of pairs of vectors
$\vecu^{\,\indt}, \vecv^{\,\indt}$ arrives for each $t\in[n]$, and the task is to compute
$({\vecu^{\,\indt}})\trans \matM \vecv^{\,\indt}$.  As before, the algorithm has
to output $({\vecu^{\,\indt}})\trans \matM \vecv^{\,\indt}$ before it gets
$\vecu^{\,\indt+1}, \vecv^{\,\indt+1}$ as input.
It is known that \OuMv is at least as difficult as \OMv.

\begin{theorem}[Theorem 2.4 in \cite{Henzinger.2015}]\label{thm:OuMv}
  If there is some $\smalleps>0$ such that \OuMv can be solved in
  $\dimn^{3-\smalleps}$ time, then the \OMvcon fails.
\end{theorem}

\subsection{The \OVcon}

While \OuMv and \OMv turn out to be suitable for Boolean CQs and the
enumeration of $k$-ary CQs, our lower bound for the counting
complexity additionally relies on the \emph{orthogonal vectors
  conjecture} 
(also known as the \emph{Boolean orthogonal detection
  problem}, see \cite{Abboud.2015,Williams.2005}). 
It is not known whether
this conjecture implies or is implied by the \OMvcon. However, it is implied
by the strong exponential time hypothesis (\SETH)
\cite{Williams.2005} and typically serves as a basis for SETH-based
lower bounds of polynomial time algorithms.

The \emph{orthogonal vectors problem} (OV) is the following static decision
problem.
Given two sets $\vecuset$ and $\vecvset$ of $n$ Boolean vectors of
dimension
$d$, decide whether there are $\vecu\in\vecuset$ and $\vecv\in\vecvset$
such that 
$\vecu\trans\vecv=0$.
This problem can clearly be solved in time 
$\bigoh(n^2d)$ by
checking all pairs of vectors, and also slightly better algorithms are
known \cite{Abboud.2015}.
The \emph{\OVcon} states that this problem cannot be solved
in truly subquadratic time if $d=\omega(\log n)$.
The exact formulation of this conjecture in terms of the parameters
varies in the literature, but all of them imply the following simple
variant which is sufficient for our purposes.

\begin{conjecture}[\OVcon]
For every $\smalleps>0$ there is no algorithm that solves \OV for $d=\lceil\log^2
n\rceil$ in time $\bigoh(n^{2-\smalleps})$.  
\end{conjecture}

\subsection{Proof Ideas}

Before we establish the lower bounds in full generality, we illustrate
the main ideas along the two representative examples
$\phione$ and $\phitwo$ defined in \eqref{eq:SxExyTy} and \eqref{eq:ExyTy}.
Note that if a conjunctive query is not \qhier, then according to
Definition~\ref{def:consquanthier} there are two distinct variables
$\varx$ and $\vary$ that do not satisfy 
one of the two
conditions.
The Boolean query $\phione$ is an example of a query where $\varx$ and
$\vary$ do not satisfy
the first condition (i.e., the condition of being hierarchical), 
and  $\phitwo$ is a query
where the quantifier-free part is hierarchical, but where $\varx$ and
$\vary$ do not satisfy the second
condition on the free variables.
Intuitively, every non-\qhier query has a subquery whose shape is
similar to either $\phione$ or $\phitwo$ 
(we will make this precise in
Section~\ref{sec:lowerbound-main}).

Let us show how the \OMvcon can be applied to obtain a lower bound for
answering the Boolean query $\phione$ under updates.

\begin{lemma}\label{lem:SxExyTy}
  Suppose there is an $\smalleps>0$ and a dynamic algorithm with
  arbitrary preprocessing time and $\updatetime=\dimn^{1-\smalleps}$ update time
  that
  answers $\phione$ in time $\answertime=\dimn^{2-\smalleps}$ on databases whose
  active domain has size $\dimn$, then \OuMv can be solved in time
  $\bigoh(n^{3-\smalleps})$. 
\end{lemma}

\begin{proof}
  We show how a query evaluation algorithm for $\phione$ can be used
  to solve \OuMv.
  We get the $\dimn\times\dimn$ matrix $\matM$ and start the
  preprocessing phase of our evaluation algorithm for $\phione$ with
  the empty database $\DB=(\relE^\DB,\relS^\DB,\relT^\DB)$ where
  $\relE^\DB=\relS^\DB=\relT^\DB=\emptyset$.
  As this database has constant size, the preprocessing phase
  finishes in constant time.
  We apply at
  most $\dimn^2$ update steps to ensure that
  $\relE^\DB=\setc{(\indi,\indj)}{\matM_{\indi,\indj}=1}$ is the
  relation corresponding to the 
  adjacency matrix $\matM$.  This
  preprocessing takes time
  $\dimn^2\updatetime=\bigoh(\dimn^{3-\smalleps})$.  If we get two
  vectors $\vecu^{\,\indt}$ and $\vecv^{\,\indt}$ in the dynamic phase of the
  \OuMv problem, we
  update $\relS^\DB$ and $\relT^\DB$ so that their 
  characteristic
  vectors agree with $\vecu^{\,\indt}$ and $\vecv^{\,\indt}$, respectively.
  Now we  answer $\phione$ on $\DB$ within time $\answertime$ and
  output $1$ if 
  $\eval{\phione}{\DB} = \Yes$ and $0$ otherwise.
  Note that by construction this answer agrees with 
  $({\vecu^{\,\indt}})\trans \matM \vecv^{\,\indt}$.
  The time of each step of the dynamic
  phase of \OuMv is bounded by $2\dimn\updatetime+\answertime=\bigoh(\dimn^{2-\smalleps})$, and the
  overall running time for \OMv accumulates to $\bigoh(\dimn^{3-\smalleps})$.
\end{proof}

Note that a lower bound on the answer time $\answertime$ of a Boolean query
directly implies the same lower bounds for the time $\countingtime$
needed to count the number of tuples and for the delay $\delaytime$ of
an enumeration algorithm.
Furthermore, this also holds true for any query that is obtained from the
Boolean query by removing quantifiers.

Now we turn to our second example $\phitwo$.
Note that the Boolean version $\exists\varx\,\phitwo(\varx)$ is \qhier
and hence can be answered in constant time under updates by
Theorem~\ref{thm:upperbound}.
Thus, a lower bound on the delay $\delaytime$ does not follow from a
corresponding lower bound on the Boolean version.
Instead we obtain the lower bound by a direct reduction from \OMv.

\begin{lemma}\label{lem:ExyTy}
  Suppose there is an $\smalleps>0$ and a dynamic algorithm with
  arbitrary preprocessing time and $\updatetime=\dimn^{1-\smalleps}$ update time
  that enumerates $\phitwo$
  with $\delaytime=\dimn^{1-\smalleps}$ delay
  on databases whose active domain has size $n$,
  then \OMv
  can be solved in time $\bigoh(\dimn^{3-\smalleps})$.
\end{lemma}

\begin{proof}
  We show that an enumeration algorithm with $\dimn^{1-\smalleps}$
  update time and $\dimn^{1-\smalleps}$ delay helps to solve \OMv in
  time $\bigoh(\dimn^{3-\smalleps})$.
  As in the proof of Lemma~\ref{lem:SxExyTy}, we 
  are given an $\dimn\times\dimn$ matrix $\matM$, start with the empty
  database $\DB=(\relE^\DB,\relT^\DB)$ where
  $\relE^\DB=\relT^\DB=\emptyset$ and perform at most $\dimn^2$ update steps to ensure that
  $\relE^\DB=\setc{(\indi,\indj)}{\matM_{\indi,\indj}=1}$.
  In the dynamic phase of \OMv,
  when a vector $\vecv^{\,\indt}$ arrives, we perform at most $\dimn$
  insertions or deletions to the relation $\relT^\DB$ such that
  $\vecv^{\,\indt}$ is the characteristic vector of $\relT^\DB$.
  Afterwards, we wait until the enumeration algorithm outputs the set
  $\phitwo(\DB)$ and output the
  characteristic vector $\vecu^{\,\indt}$ of this set.
  By construction we have $\vecu^{\,\indt}=\matM \vecv^{\,\indt}$.
  If the enumeration algorithm has
  update time $\updatetime$ and delay $\delaytime$, then the overall
  running time of this step is bounded by
  $\dimn\updatetime+\dimn\delaytime$ which by the assumptions of our
  lemma is bounded by $\bigoh(n^{2-\smalleps})$.  Hence, the overall running
  time for solving the \OMv is bounded by $\bigoh(n^{3-\smalleps})$.
\end{proof}

Finally, we consider the counting problem for $\phitwo$. Again, we
cannot reduce from its tractable Boolean version. 
Moreover, we were not able to use \OMv
directly, in a similar way as in the proof of the previous lemma.
Instead, we reduce from the orthogonal vectors problem.

\begin{lemma}\label{lem:ExyTy-count}
  Suppose there is an $\smalleps>0$ and a dynamic algorithm with
  arbitrary preprocessing time 
  and $\updatetime=\dimn^{1-\smalleps}$ update time  
  that computes
  $\Setsize{\eval{\phitwo}{\DB}}$ in time 
  $\countingtime=\dimn^{1-\smalleps}$ 
  on databases whose active domain has size $n$,
  then the \OVcon fails.
\end{lemma}

\begin{proof}
  As in the previous proof we assume that there is a dynamic counting
  algorithm for $\phitwo$ and start its preprocessing with the empty database over the
  schema $\{\relE,\relT\}$.
  Afterwards, we use at most $nd$ updates (where $d=\lceil\log^2(n)\rceil$) to
  encode all $d$-dimensional vectors $\vecu^{\,1}$, \ldots, $\vecu^{\,n}$ in $\vecuset$ into the binary relation
  $\relE^\DB\subseteq [n]\times [d]$ such that
  $(i,j)\in E^\DB$ if and only if the $j$-th component of $\vecu^{\,i}$ is
  $1$.
  Then we make at most $d$ updates to $T^\DB$ to ensure that the first
  vector $\vecv^{\,1}\in\vecvset$ is the 
  characteristic vector of $T^\DB$.
  Now we compute $\Setsize{\eval{\phitwo}{\DB}}$. Note that $\Setsize{\eval{\phitwo}{\DB}}<n$ if and
  only if $(\vecu^{\,i})\trans\vecv^{\,1}=0$ for some $i\in[n]$.
  If this is the case, we output that there is a pair of orthogonal
  vectors.
  Otherwise, we know that $\vecv^{\,1}$ is not orthogonal to 
  any
  $\vecu^{\,i}$
  and apply the same procedure for $\vecv^{\,2}$, which requires again at
  most $d$ updates to $T^\DB$ and one call of the $\COUNT$ routine.
  Repeating this procedure for all $n$ vectors in $\vecvset$ takes
  time $\bigoh\big(nd\updatetime+n(d\updatetime+\countingtime)\big)\leq \bigoh(n^{2-\smalleps/2})$ and
  solves \OV in subquadratic time. 
\end{proof}

\subsection{Proofs of the Main Theorems}\label{sec:lowerbound-main}

\label{sec:lowerbound-Boolean}

In this section we prove our lower bound
Theorems~\ref{thm:enumerating_CQ_intro},
\ref{thm:modelchecking_joins_intro}, and
\ref{thm:counting_CQ_lowerbound}.

We will use standard notation concerning homomorphisms
(cf., e.g.\ \cite{AHV-Book}). In particular, for CQs $\queryphi$ and
$\queryphialt$ we will write
$\homh:\queryphi\to\queryphialt$ to indicate that $\homh$ is a homomorphism from
$\queryphi$ to $\queryphialt$ (as defined in
Section~\ref{sec:mainresults}).
A homomorphism $\homg:\DB\to\queryphi$ from a database $\DB$ to a CQ
$\queryphi$ is a mapping from $\Adom(\DB)$ to $\Vars(\queryphi)$ such that
whenever $(a_1,\ldots,a_r)$ is a tuple in some relation $\relR^\DB$ of
$\DB$, then $\relR \homg(a_1)\cdots \homg(a_r)$ is an atom of $\queryphi$.
A homomorphism $\homh:\queryphi\to\DB$ from a CQ $\queryphi$ to a database
$\DB$ is a mapping from $\Vars(\queryphi)$ to $\Adom(\DB)$ such that 
whenever $\relR \varu_1\cdots \varu_\arityr$ is an atom of $\queryphi$, then
$\big(\homh(\varu_1), \ldots, \homh(\varu_\arityr)\big)\in \relR^{\DB}$. 
Obviously, for a $k$-ary CQ $\queryphi(x_1,\ldots,x_k)$ and a database
$\DB$ we have
$\eval{\queryphi}{\DB}=\setc{(\homh(x_1),\ldots,\homh(x_k))}{\homh\text{
  is a homomorphism from $\queryphi$ to $\DB$}}$.

We first generalise the proof idea of Lemma~\ref{lem:SxExyTy} to
all Boolean conjunctive queries $\queryphi$ that
do not
satisfy the requirement \eqref{item:hier-cond} of
Definition~\ref{def:consquanthier}.
Thus assume that there
are two variables $\varx, \vary\in\Vars(\query)$ and three atoms
$\sgpsix, \sgpsixy, \sgpsiy$ of $\query$ 
with
$\Vars(\sgpsix)\cap\set{\varx,\vary}=\set{\varx}$,
$\Vars(\sgpsixy)\cap\set{\varx,\vary}=\set{\varx,\vary}$, and
$\Vars(\sgpsiy)\cap\set{\varx,\vary}=\set{\vary}$.  Without loss of
generality we assume that
$\Vars(\queryphi)=\set{\varx,\vary,\varz_1,\ldots,\varz_\ell}$.
For a
given $\dimn\times\dimn$ matrix $\matM$
we fix a domain $\domn$ that consists of $2n+\ell$ elements  
$\setc{\verta_\indi, \vertb_\indi}{\indi\in[\dimn]} \cup \setc{\vertc_\inds}{\inds\in[\ell]}$.
For $\indi,\indj\in[\dimn]$ we let
$\iotasubij$
be the injective mapping
from $\Vars(\queryphi)$ to $\domn$
with
$\iotasubij(\varx)=\verta_\indi$, 
$\iotasubij(\vary)=\vertb_\indj$, and
$\iotasubij(\varz_\inds)=\vertc_\inds$ for all $s\in[\ell]$.

For the matrix $\matM$ and for
$\dimn$-dimensional vectors $\vecu$ and $\vecv$, 
we define a $\sigma$-db
$\DB=\DB(\queryphi,\matM,\vecu,\vecv)$ with 
$\adom{\DB}\subseteq\domn$ as follows (recall our notational
convention that $\vecu_\indi$ denotes the $\indi$-th entry of a vector
$\vecu$).
For every atom
$\sgpsi=\relR\varw_1\cdots\varw_\arityr$ in $\queryphi$ we include
in $\relR^\DB$ the tuple
$\big(\iotasubij(\varw_1),\ldots,\iotasubij(\varw_\arityr)\big)$
\begin{itemize}
\item
for all $\indi,\indj\in[\dimn]$ such that 
 $\vecu_\indi=1$, \
 if $\sgpsi=\sgpsix$,  
\item
for all $\indi,\indj\in[\dimn]$ such that 
$\vecv_\indj=1$, \ 
if $\sgpsi=\sgpsiy$, 
\item
for all $\indi,\indj\in[\dimn]$ such that $\matM_{\indi,\indj}=1$, \ if
$\sgpsi=\sgpsixy$, \ and
\item 
for all $\indi,\indj\in[\dimn]$, \ if $\sgpsi\notin\set{\sgpsix,\sgpsixy,\sgpsiy}$. 
\end{itemize}

Note that the relations in the atoms $\sgpsix$, $\sgpsiy$, and
$\sgpsixy$ are used to encode $\vecu$, $\vecv$, and $\matM$,
respectively.  
Moreover, 
since $\sgpsix$ ($\sgpsiy$) does not contain the variable $\vary$ ($\varx$),
two databases $\DB=\DB(\queryphi,\matM,\vecu,\vecv)$ and
$\DB'=\DB(\queryphi,\matM,\mbox{$\vecu\,{}'$},\mbox{$\vecv\,{}'$})$ differ only in at most
$2\dimn$ tuples.
Therefore, $\DB'$ can be obtained from $\DB$ by $2\dimn$ update steps. 
It follows from the definitions that $\iotasubij$ is a homomorphism
from $\queryphi$ to $\DB$ if and only if $\vecu_\indi=1$,
$\vecv_\indj=1$, and $\matM_{\indi,\indj}=1$.  Therefore,
$\vecu\trans \matM \vecv = 1$ if and only if there are
$\indi,\indj\in[n]$ such that $\iotasubij$ is a
homomorphism from $\queryphi$ to $\DB$.

We let
$\homDBtoquery$ be the (surjective) mapping from $\domn$ to
$\Vars(\queryphi)$ 
defined by $\homDBtoquery(\vertc_\inds)\defi\varz_\inds$,
$\homDBtoquery(\verta_\indi)\defi\varx$, and
$\homDBtoquery(\vertb_\indj)\defi\vary$ for all $\indi,\indj\in[n]$
and $\inds\in[\ell]$.
Clearly, $\homDBtoquery$ is a
homomorphism from $\DB$ to $\queryphi$.
Obviously, the following is true for every mapping
$\homh$ from $\Vars(\queryphi)$ to $\Adom(\DB)$ and for all $\varw\in\Vars(\queryphi)$:
  if $\homh(w)=\vertc_\inds$ for some $\inds\in[\ell]$, 
  then $(\homDBtoquery\circ h)(w)=\varz_\inds$;
  if $\homh(w)=\verta_\indi$ for some $\indi\in[n]$, 
  then $(\homDBtoquery\circ h)(w)=\varx$;
  if $\homh(w)=\vertb_\indj$ for some $\indj\in[n]$, 
  then $(\homDBtoquery\circ h)(w)=\vary$.

We define the partition
$\partPfull=\Set{\set{\vertc_1},\ldots,\set{\vertc_\ell},\setc{\verta_\indi}{\indi\in[\dimn]},\setc{\vertb_\indj}{\indj\in[\dimn]}}$
of $\domn$ and say that a mapping
$\homh$ from $\Vars(\queryphi)$ to $\Adom(\DB)$
\emph{respects} $\partPfull$,
if for each set from the partition there is exactly one element in the
image of $\homh$.
\begin{claim}\label{claim:respectinghom_uMv}
  $\vecu\trans \matM \vecv = 1$ \ $\iff$ \
  There exists a homomorphism $\homh\colon\queryphi\to \DB$ that
  respects $\partPfull$.
\end{claim}

{%
\begin{proof}
  For one direction assume that $\vecu\trans \matM \vecv = 1$.  Then
  there are $\indi,\indj\in[\dimn]$ such that $\iotasubij$ is a homomorphism
  from $\phi$ to $\DB$
  that respects $\partPfull$.  For the other direction assume that
  $\homh\colon\queryphi\to \DB$ is a homomorphism that respects
  $\partPfull$.  It follows that $(\homDBtoquery\circ\homh)$ is a
  bijective homomorphism from $\queryphi$ to $\queryphi$.
  Therefore, 
  it can easily be verified that
  $\homh\circ(\homDBtoquery\circ\homh)^{-1}$ is a homomorphism from
  $\queryphi$ to $\DB$ which equals $\iotasubij$ for some
  $\indi,\indj\in[\dimn]$.  This implies that $\vecu\trans \matM \vecv = 1$.
\end{proof}
}
\begin{claim}\label{claim:respectinghom_core}
If  $\queryphi$ is a core, then every homomorphism $\homh\colon\queryphi\to \DB$ respects $\partPfull$. 
\end{claim}
{%
\begin{proof}
  Assume for contradiction that $\homh\colon\queryphi\to \DB$ is a
  homomorphism that does not respect $\partPfull$.  Then
  $(\homDBtoquery\circ\homh)$ is a homomorphism from $\queryphi$ into a
  proper subquery of $\queryphi$, contradicting that $\queryphi$
  is a core.
\end{proof}
 }

\begin{proof}[Proof of Theorem~\ref{thm:modelchecking_joins_intro}]
  Assume for contradiction that \allowbreak the
  query answering problem
  for $\queryphi$ and hence for its non-\qhier
  core $\queryphicore$ can be solved with update time
  $\updatetime=\bigoh(\actdomsize^{1-\smalleps})$ 
  and answer time $\answertime=\bigoh(\actdomsize^{2-\smalleps})$.
  We can use this
  algorithm to solve \OuMv in time
  $\bigoh(\actdomsize^{3-\smalleps})$ as follows.

  In the
  preprocessing phase, we are given the $\dimn\times\dimn$ matrix
  $\matM$ and let $\vecu^{\,0}$, $\vecv^{\,0}$ be the all-zero vectors
  of dimension $\dimn$.
  We start the preprocessing phase of our evaluation algorithm for
  $\queryphicore$ with the empty database. As this database has constant size,
  the preprocessing phase finishes in constant time.
  Afterwards, we use at most $\dimn^2$ insert operations to build the
  database $\DB(\queryphicore,\matM,\vecu^{\,0},\vecv^{\,0})$.
  All this is done within time $\bigOh(n^2\updatetime)=\bigOh(n^{3-\epsilon})$.

  When a pair
  of vectors $\vecu^{\,\indt}$, $\vecv^{\,\indt}$ (for $t\in[\dimn]$) arrives, we change the
  current database \allowbreak
  $\DB(\queryphicore,\matM,\vecu^{\,\indt-1},\vecv^{\,\indt-1})$ into 
  $\DB(\queryphicore,\matM,\vecu^{\,\indt},\vecv^{\,\indt})$ by using
  at most $2\dimn$ update steps.  By the Claims~\ref{claim:respectinghom_uMv}
  and \ref{claim:respectinghom_core} we know that
  $(\vecu^{\,\indt})\trans \matM \vecv^{\,\indt} = 1$ if, and only if, there is
  a homomorphism from $\queryphicore$ to
  $\DB \deff \DB(\queryphicore,\matM,\vecu^{\,\indt},\vecv^{\,\indt})$.
  Hence, after answering the Boolean query $\queryphicore$ on
  $\DB$ in time $\answertime=\setsize{\adom{\DB}}^{2-\smalleps}= \bigoh(\dimn^{2-\smalleps})$
  we can output the value of $(\vecu^{\,\indt})\trans \matM \vecv^{\,\indt}$.
  The time of each step of the dynamic
  phase of \OuMv is bounded by
  $2\dimn\updatetime+\answertime=\bigoh(\dimn^{2-\smalleps})$.
  Thus, the
  overall running time sums up to $\bigoh(\dimn^{3-\smalleps})$, contradicting
  the \OMvcon by Theorem~\ref{thm:OuMv}.
\end{proof}
The same reduction from \OuMv to the query evaluation problem for
conjunctive queries is also useful for the lower bound on the
counting problem, provided that the query is not hierarchical.
If the query is hierarchical, but the quantifiers are not in the
correct form 
(such that the query is not \qhier), then \OuMv does not
provide us with the desired lower bound proof
and we have to stick to the \OVcon{} instead.
Another crucial difference between the Boolean and the non-Boolean
case is the following: the 
dynamic counting problem for the Boolean query 
$\query=\exists\varx\eexists\vary\body{\relE\varx\varx\und
\relE\varx\vary \und \relE\vary\vary}$ is easy (because its
core is $\exists\varx \,\relE\varx\varx$), whereas the 
dynamic counting problem for its non-Boolean version 
$\query(\varx,\vary)=\body{\relE\varx\varx\und \relE\varx\vary\und
\relE\vary\vary}$ 
is hard (because the query is a non-\qhier core).
To take care of this phenomenon we utilise the following lemma.

\newcommand{\TextOfLemmaSimulatingUnaryRelations}{
  Suppose that $\query(x_1,\ldots,x_k)$ is a conjunctive query where
  $\Setsize{\query(\DB)}$ can be computed with $\preprocessingtime$
  preprocessing time, $\updatetime$ update time
  $\countingtime$ counting time.
  Let $\setX_{\varx_1},\ldots,\setX_{\varx_k}$ be pairwise disjoint subsets
  of the domain.
  Then the number
  $\Setsize{\query(\DB) \cap
    \bigl(\setX_{\varx_1} \times \cdots \times
    \setX_{\varx_{k}}\bigr)}
  $
  can be computed with
  $2^{\bigoh(k)}\updatetime$
  update time after $2^{\bigoh(k)}\preprocessingtime$ preprocessing.
  }
\renewcommand{\TextOfLemmaSimulatingUnaryRelations}{
  Let $\query(\varx_1,\ldots,\varx_k)$ be a conjunctive query
  and $\setX_{\varx_1},\ldots,\setX_{\varx_k}$ be pairwise disjoint subsets
  of $\Dom$. 
  Suppose that for every database $\DB$ under consideration there is a
  homomorphism $\homg\colon \DB\to \query$ such that
  $\homg(\setX_{\varx_i})=\set{\varx_i}$ for all $i\in[k]$. 
  If
  $\Setsize{\query(\DB)}$ can be computed with $\preprocessingtime$
  preprocessing time, $\updatetime$ update time, 
  and $\countingtime$ counting time,
  then the number
  $\Setsize{\query(\DB) \cap
    \bigl(\setX_{\varx_1} \times \cdots \times
    \setX_{\varx_{k}}\bigr)}
  $
  can be computed with
  $2^{\bigoh(k)}(\updatetime+\countingtime)$ update time 
  and $\bigOh(1)$ counting time 
 after 
  $2^{\poly(\query)}+2^{\bigoh(k)}(\preprocessingtime+\countingtime)$ preprocessing.
  }

\begin{lemma}\label{lem:simulating_unary_relations}
  \TextOfLemmaSimulatingUnaryRelations{}
\end{lemma}

In the static setting, a similar result 
was shown by Chen and Mengel (see Section~7.1 in \cite{Chen.2015}) and it turns out
that our dynamic version can be proven using
the same techniques.
We remark that the lemma holds even if we drop the additional
requirement on the existence of the homomorphism $\homg$.
However, as the databases we construct in our lower bound proof have the desired
structure, this additional requirement helps to simplify
the proof.

\begin{proof}[Proof of Lemma~\ref{lem:simulating_unary_relations}]
We first reduce the given task to
counting tuples up to permutations,
that is, computing the size of the set  
$\setcalR(\DB) \deff$
\[
\Setc{(\verta_1,\ldots,\verta_k)\in\query(\DB)}{
  \text{there is a permutation $\perm\colon [k]\to [k]$ 
        with $\verta_i\in\setX_{\varx_{\perm(i)}}$ for all $i\in[k]$}}\,.
\]
Let $\Pi$ be the set of all permutations $\perm\colon [k]\to [k]$
such that the mapping $\big(\varx_i\mapsto \varx_{\perm(i)}\big)_{i\in[k]}$ extends to an
endomorphism on $\query$ 
(i.e., a homomorphism from $\query$ to $\query$).
We now show that
\begin{equation}
  \label{eq:4}
  \Setsize{
    \query(\DB) \cap\bigl(\setX_{\varx_1} \times \cdots \times
    \setX_{\varx_{k}}\bigr)
   }
   \,\cdot\,
   \Setsize{\Pi}
   \ \ = \ \ 
   \setsize{\setcalR(\DB)}\,.
\end{equation}
First note that if $(\verta_1,\ldots,\verta_k)\in\query(\DB) \cap
    \bigl(\setX_{\varx_1} \times \cdots \times
    \setX_{\varx_{k}}\bigr)$, then
    $(\verta_{\perm(1)},\ldots,\verta_{\perm(k)})\in\setcalR(\DB)$ for
    all $\perm\in\Pi$. Thus, the $\leq$-direction of \eqref{eq:4} follows because all $\setX_{\varx_i}$ are pairwise
    disjoint.
    For the other direction,
    consider an arbitrary tuple
    $(\verta_{1},\ldots,\verta_{k})\in\setcalR(\DB)$. In particular,
    there is a permutation $\perm:[k]\to[k]$ such that
    $\verta_i\in\setX_{x_{\perm(i)}}$ for all $i\in[k]$.
    Furthermore, since
    $(\verta_{1},\ldots,\verta_{k})\in\eval{\query}{\DB}$, there is a
    homomorphism $\homh\colon\query\to\DB$ with $\homh(\varx_\indi)=\verta_\indi$
    for all $\indi\in[k]$.
    When combining $\homh$ with the
    homomorphism $\homg\colon\DB\to\query$ given by the lemma's
    assumption, we obtain the endomorphism $f\deff(\homg\circ\homh)$ on
    $\phi$ that satisfies
    $f(\varx_\indi)=\varx_{\pi(i)}$ for all
    $\indi\in[k]$.
    Thus, $\pi\in\Pi$. To conclude the proof of the $\geq$-direction
    of \eqref{eq:4}, it remains to show that
    $(\verta_{\pi^{-1}(1)},\ldots,\verta_{\pi^{-1}(k)})\in
    \eval{\phi}{\DB}$, i.e., it remains to show that there is a
    homomorphism $\homh':\query\to\DB$ with
    $\homh'(\varx_\indi)=\verta_{\pi^{-1}(\indi)}$ for all $\indi\in[k]$.
    Since $\pi$ is a permutation, $\pi^{-1}=\pi^m$ for some $m\geq 1$.
    Thus, iterating $f$ for $m$ times yields the endomorphism $f^m$ on
    $\phi$ with
    $f^m(\varx_\indi)=\varx_{\pi^m(i)}=\varx_{\pi^{-1}(i)}$ for all
    $i\in[k]$. Therefore, choosing $h'\deff (h\circ f^m)$
    completes the proof of \eqref{eq:4}.

    As the set $\Pi$ depends only on the query and can be computed in time
    $2^{\poly(\query)}$ in the preprocessing phase, it suffices
    to store and update information on the number
    $\setsize{\setcalR(\DB)}$, whenever the database $\DB$ is updated.
    To do this efficiently, we store for every $\setI\subseteq [k]$ and every
    $\indj\in \set{0,\ldots,k}$
    the sizes of the auxiliary sets 
    \begin{equation}
      \label{eq:5}
      \setcalR_{\setI,\indj} \ \ \deff \ \
      \setc{\,(\verta_1,\ldots,\verta_k)\in\query(D)}{\Setsize{\set{\verta_1,\ldots,\verta_k}\cap
      \bigl(\textstyle\bigcup_{\indi\in\setI}\setX_{\varx_\indi}\bigr)}=\indj\,}.
    \end{equation}
    Note that
    \begin{equation*}
    \setcalR(\DB) \quad = \quad
    \setcalR_{[k],k} \ \setminus \
    \bigcup_{i\in [k]} \setcalR_{[k]\setminus\set{i},k}\,.
    \end{equation*}
    Hence, we can use the cardinalities of the sets
    $\setcalR_{\setI,\indj}$ (with index $j\deff k$)
    to compute $\setsize{\setcalR(\DB)}$ by the
    following application of the inclusion-exclusion principle
    \begin{align}
      \label{eq:4Strich}
      \setsize{\setcalR(\DB)}
      &\ \ = \ \ \setsize{\setcalR_{[k],k}} -
      \Bigl|\bigcup_{\indi\in[k]}\setcalR_{[k]\setminus\set{\indi},k}\Bigr|
      \ \ = \ \ \setsize{\setcalR_{[k],k}} - \sum_{\emptyset\neq
        \setI\subseteq [k]} (-1)^{\setsize{\setI}-1}\cdot
      \Bigl|\bigcap_{\indi\in\setI}\setcalR_{[k]\setminus\set{\indi},k}\Bigr|\\
      & \label{eq:5}
      \ \ = \ \ \sum_{
        \setI\subseteq [k]} (-1)^{\setsize{\setI}}\cdot
      \Bigl|\setcalR_{[k]\setminus\setI,k}\Bigr|.
    \end{align}
    In
    order to compute the numbers $\setsize{\setcalR_{\setI,\indj}}$
    efficiently, we
    consider for every $\setI\subseteq [k]$ and every 
    $\ell\in [k]$ the
    database $\DB_{\setI,\ell}$ which is obtained from $\DB$ by
    replacing every element from $\bigcup_{\indi\in\setI}\setX_{\varx_\indi}$
    by $\ell$ copies of itself. Precisely, we consider fresh
    elements $\langle\verta\rangle^1,\ldots,\langle\verta\rangle^k$
    for every $\verta$ from the domain and define for every $r$-ary
    $\relR\in\schema$
  \begin{equation*}
    \label{eq:6}
    R^{\DB_{\setI,\ell}} \deff
    \Set{\bigl(\langle\verta_1\rangle^{\inds_1},\ldots,\langle\verta_r\rangle^{\inds_r}\bigr)
      \colon
      (\verta_1,\ldots,\verta_r)\in
    R^{\DB}, \inds_\indi = 1\text{ for }\indi\notin\setI,\text{ and }\inds_\indi
    \in [\ell] \text{ for }\indi\in\setI}.
\end{equation*}
We maintain these 
$k2^k$ auxiliary databases $\DB_{\setI,\ell}$ in parallel and call the $\COUNT$ routines to determine the new
result sizes $\Setsize{\query(\DB_{\setI,\ell})}$.
All this can be done in time
$(\updatetime+\countingtime) 2^{\bigoh(k)}$.
Afterwards, we compute the numbers
$\Setsize{\setcalR_{\setI,\indj}}$ given the cardinalities
$\Setsize{\query(\DB_{\setI,\ell})}$ as follows.
For every $\ell\in[k]$ we have
\begin{equation*}
  \label{eq:7}
  \Setsize{\query(\DB_{\setI,\ell})} 
  \quad = \quad
  \sum_{\indj=0}^k \ \ell^\indj \cdot \Setsize{\setcalR_{\setI,\indj}}.
\end{equation*}
Hence, in order to compute the values of
$\Setsize{\setcalR_{\setI,\indj}}$ 
it suffices to solve $2^k$ systems of linear
equations of the form $V\vec x = \vec b$, where $V$ is a $k\times (k+1)$ Vandermonde
matrix with $V_{\ell,\indj}\deff\ell^\indj$ (for all $\ell\in[k]$ and $\indj\in\set{0,\ldots,k}$).
Therefore, all these values are uniquely determined and %
  can be computed in time $2^{\bigoh(k)}$.
Finally, we can use \eqref{eq:5} to compute and store the desired
value $\setsize{\setcalR(\DB)}$.
Note that the entire time used for performing an update is 
$\bigOh((\updatetime+\countingtime)2^{\bigOh(k)})$, as claimed.
\end{proof}

\begin{proof}[Proof of Theorem~\ref{thm:counting_CQ_lowerbound}]
  Let $\varphi$ be the non-\qhier homomorphic core of the given conjunctive query.
  For contradiction, assume that the counting problem for $\varphi$ can be solved with 
  update time $\updatetime=\bigOh(n^{1-\epsilon})$ and counting time $\countingtime=\bigOh(n^{1-\epsilon})$.

  We first handle the case where $\varphi$ does not satisfy the first condition \eqref{item:hier-cond}
  of Definition~\ref{def:consquanthier}.  Hence, as in the Boolean
  case, there are variables $\varx,\vary$ and atoms $\psix$,
  $\psixy$, $\psiy$ with
  $\Vars(\psix)\cap\set{\varx,\vary}=\set{\varx}$,
  $\Vars(\psixy)\cap\set{\varx,\vary}=\set{\varx,\vary}$, and
  $\Vars(\psiy)\cap\set{\varx,\vary}=\set{\vary}$.  Again we reduce
  from \OuMv, assume that
  $\Vars(\queryphi)=\set{\varx,\vary,\varz_1,\ldots,\varz_\ell}$, and
  let
  $\DB=\DB(\queryphi,\matM,\vecu,\vecv)$ be the $\sigma$-db defined above
  for a given $\dimn\times\dimn$ matrix $\matM$ and
  $\dimn$-dimensional vectors $\vecu$ and $\vecv$.

  Note that the lower bound for counting the query result does not follow from the lower bound for Boolean queries
  (Theorem~\ref{thm:modelchecking_joins_intro}), because it might be
  the case that the core of the Boolean version
  $\exists\varx\eexists\vary\eexists\varz_1\cdots\eexists\varz_\ell\,\query$
  of $\query$ actually \emph{is} \qhier (consider
  $\bodyjoin{\relE\varx\varx\und \relE\varx\vary\und \relE\vary\vary}$ for example).
  To take care of this, we apply
  Lemma~\ref{lem:simulating_unary_relations} and let 
  $\setX_{\varz_1}\deff\set{\vertc_1}$,
  \ldots,
  $\setX_{\varz_\ell}\deff\set{\vertc_\ell}$,
  $\setX_\varx\deff\setc{\verta_\indi}{\indi\in[\dimn]}$,
  $\setX_\vary\deff\setc{\vertb_\indi}{\indi\in[\dimn]}$
  be the corresponding sets
  that partition the domain of $\DB(\queryphi,\matM,\vecu,\vecv)$.
  Let us say that a homomorphism $\homh\colon\queryphi\to \DB$
  is \emph{\good} if $\homh(\varw)\in\setX_{\varw}$ for all
  $\varw\in\free(\query)$.
  As in the proof of Claim~\ref{claim:respectinghom_core}, we
  have that every \good homomorphism respects $\partPfull$ since
  otherwise 
  $(\homDBtoquery\circ\homh)$ would be a homomorphism from $\queryphi$ into a
  proper subquery of $\queryphi$, contradicting that $\queryphi$
  is a core.
  By Lemma~\ref{lem:simulating_unary_relations} we can count the
  number of 
  result tuples produced by good homomorphisms $h:\phi\to\DB$, and this can be done with counting time $\bigOh(1)$ and
  update time $\bigOh(n^{1-\epsilon})$ 
  (here, the factor $2^{\bigOh(k)}$ mentioned in the lemma is subsumed by the $\bigOh$-notation,
  since the query $\queryphi$ is fixed).
  By Claim~\ref{claim:respectinghom_uMv}, the number of result tuples produced by good homomorphisms $h:\phi\to\DB$ is
  $>0$ if, and only if, $\vecu\trans \matM \vecv = 1$.
  Thus, we can proceed in a similar way as in the proof of Theorem~\ref{thm:modelchecking_joins_intro}
  and solve 
  \OuMv within time $\bigOh(n^{3-\epsilon})$, contradicting the \OMvcon by
  Theorem~\ref{thm:OuMv}.

  For the second case suppose that $\query$ does not satisfy
  condition \eqref{item:quant-cond} of
  Definition~\ref{def:consquanthier}. Thus assume that there are two
  variables $\varx\in\free(\query), \vary\in\Vars(\query)\setminus\free(\query)$ and two atoms
  $\sgpsixy, \sgpsiy$ of $\query$ with
  $\Vars(\sgpsixy)\cap\set{\varx,\vary}=\set{\varx,\vary}$ and
  $\Vars(\sgpsiy)\cap\set{\varx,\vary}=\set{\vary}$. Without
  loss of generality we assume that $\query = \query(\varx,\varz_1,\dots,\varz_{\ellalt})$
  and $\Vars(\queryphi)=\set{\varx,\vary,\varz_1,\ldots,\varz_\ell}$
  for some $\ellalt\leq\ell$.  We
  reduce from \OV, generalising our example in
  Lemma~\ref{lem:ExyTy-count}.  Suppose that
  $\vecuset=\Set{\vecu^{\,1},\ldots,\vecu^{\,\dimn}}$ and
  $\vecvset=\Set{\vecv^{\,1},\ldots,\vecv^{\,\dimn}}$ are two sets of $n$
  Boolean vectors, each of length $d=\lceil\log^2 n\rceil$.  
  We fix a domain 
  $\domn$ that consists of $n+d+\ell$ elements $\setc{a_i}{i\in[n]}\cup\setc{b_j}{j\in[d]}\cup\setc{c_s}{s\in\ell}$, 
  and we let
  $\setX_\varx\deff\setc{\verta_\indi}{\indi\in[\dimn]}$,
  $\setX_\vary\deff \setc{\vertb_\indj}{\indj\in[d]}$, and
  $\setX_{\varz_\inds}\deff\set{\vertc_\inds}$ for all
  $\inds\in[\ell]$.
  As before,
  for $(\indi,\indj)\in [\dimn]\times[d]$ we let $\iotasubij$ be the
  injective mapping from $\Vars(\queryphi)$ to $\domn$ with
  $\iotasubij(\varx)=\verta_\indi$, 
  $\iotasubij(\vary)=\vertb_\indj$, and
  $\iotasubij(\varz_\inds)=\vertc_\inds$ for all $\inds\in[\ell]$.

  For each vector $\vecv\in \vecvset$ we define a $\sigma$-db
  $\DB=\DB(\queryphi,\vecuset,\vecv)$ with
  $\adom{\DB}\subseteq\domn$ as follows.
  For every atom $\sgpsi=\relR\varw_1\cdots\varw_\arityr$ in
  $\queryphi$ we include in $\relR^\DB$ the tuple
  $\big(\iotasubij(\varw_1),\ldots,\iotasubij(\varw_\arityr)\big)$
  \begin{mi}
  \item for all $(\indi,\indj)\in [\dimn]\times[d]$ such that the
    $\indj$-th component of $\vecu^{\,\indi}$ is 1, \ if
    $\sgpsi=\sgpsixy$, 
  \item for all $(\indi,\indj)\in [\dimn]\times[d]$ such that
    the
    $\indj$-th component of $\vecv$ is 1, \ if
    $\sgpsi=\sgpsiy$, \ and
  \item for all $(\indi,\indj)\in [\dimn]\times[d]$, \ if $\sgpsi\notin\set{\sgpsixy,\sgpsiy}$.
  \end{mi}
  From this definition it follows that
  \begin{equation}
    \label{eq:3}
    \query(\DB)\ \cap\
    \bigl(\setX_{\varx} \times \setX_{\varz_1} \times \cdots \times
    \setX_{\varz_{\ellalt}}\bigr)
    \ \ = \ \
    \Setc{\,(\verta_\indi,\vertc_1,\ldots,\vertc_{\ellalt})}{
        \indi\in[\dimn], \ (\vecu^{\,\indi})\trans\vecv \neq 0 \, }
  \end{equation}
  and hence the size of this set equals the number of vectors from
  $\vecuset$ that are non-orthogonal to $\vecv$.
  By Lemma~\ref{lem:simulating_unary_relations} we can 
  count
  this number using $\bigOh(1)$ counting time and $\bigOh(n^{1-\epsilon})$ update time.

  In order to solve \OV, we first apply 
  the preprocessing phase for the empty database and then perform
  $\bigoh(\dimn d)$ update steps
  to build $\DB(\query,\vecuset,\vecv^{\,1})$.
  Afterwards, we compute the number of vectors in $\vecuset$ that are
  non-orthogonal to $\vecv^{\,1}$. If this number is $<n$, we
  know that there is an orthogonal pair of vectors. Otherwise, we apply
  at most $d$ update steps to the relation of the atom $\psiy$ in order to obtain the database
  $\DB(\query,\vecuset,\vecv^{\,2})$ and check again if $\vecv^{\,2}$ is
  orthogonal to some vector in $\vecuset$.
  We repeat this procedure for all vectors in $\vecvset$ and this
  allows us to solve \OV in time 
  $\bigOh(ndn^{1-\epsilon})=
  \bigoh(\dimn^{2-\smalleps/2})$, 
  contradicting the \OVcon{}.
\end{proof}

We now prove the hardness result
for enumerating the results of self-join free conjunctive queries.

\begin{proof}[Proof of Theorem~\ref{thm:enumerating_CQ_intro}] 
If $\queryphi$ does not satisfy condition \eqref{item:hier-cond}
of Definition~\ref{def:consquanthier}, then 
also the query's Boolean version 
$\exists\varx_1\cdots\exists\varx_k\,\queryphi$ (where
$\free(\queryphi)=\set{\varx_1,\ldots,\varx_k}$) is non-\qhier; and this Boolean version is its own core, since
$\queryphi$ is self-join free.
The lower bound now follows immediately from
Theorem~\ref{thm:modelchecking_joins_intro}, because if the results of 
$\queryphi$ can be enumerated with preprocessing time $\preprocessingtime$, update time $\updatetime$, and 
delay $\delaytime$, then we can 
answer $\exists\varx_1\cdots\exists\varx_k\,\queryphi$ with preprocessing time $\preprocessingtime$, 
update time $\updatetime$, and answer time
$\answertime = \delaytime$: to answer the Boolean query, just start enumerating the answers of $\queryphi$
and answer $\Yes$ or $\No$ depending on whether the first output is a tuple or
the \EOE message.

  We can therefore
  assume that $\queryphi$ satisfies condition \eqref{item:hier-cond},
  but not \eqref{item:quant-cond} of Definition~\ref{def:consquanthier}.
  Hence there are a free
  variable $\varx$, a quantified variable $\vary$, and two atoms
  $\sgpsixy$, $\sgpsiy$ such that
  $\Vars(\sgpsixy)\cap\set{\varx,\vary}=\set{\varx,\vary}$ and
  $\Vars(\sgpsiy)\cap\set{\varx,\vary}=\set{\vary}$.  Without loss of
  generality we let
  $\Vars(\queryphi)=\set{\varx,\vary,\varz_1,\ldots,\varz_\ell}$ and
  $\free(\queryphi)=\set{\varx,\varz_1,\ldots,\varz_\ellalt}$ and
  consider the query $\queryphi(\varx,\varz_1,\ldots,\varz_\ellalt)$.

  For contradiction, assume that there is a dynamic algorithm that enumerates $\queryphi$
  with $\bigoh(\actdomsize^{1-\smalleps})$ delay and
  $\bigoh(\actdomsize^{1-\smalleps})$ update time.
  We want to use this
  algorithm to solve \OMv in time $\bigoh(n^{3-\smalleps})$.  For this
  we encode an $\dimn\times\dimn$ matrix $\matM$ and an
  $\dimn$-dimensional vector $\vecv$ into a 
  database, in a
  similar way as we have done for \OuMv:  We
  define the $\sigma$-db $\DB=\DB(\queryphi,\matM,\vecv)$ over
  the domain
  $\domn=\setc{\verta_\indi, \vertb_\indi}{\indi\in[\dimn]} \cup
  \setc{\vertc_\inds}{\inds\in[\ell]}$ consisting of $2n+\ell$
  elements. 
  For $\indi,\indj\in[\dimn]$ we let
  $\iotasubij:\Vars(\queryphi)\to\domn$ be the injective mapping that
  sets 
  $\iotasubij(\varx)=\verta_\indi$, 
  $\iotasubij(\vary)=\vertb_\indj$, and
  $\iotasubij(\varz_\inds)=\vertc_\inds$ for all $\inds\in[\ell]$.
  For every atom
  $\sgpsi=\relR\varw_1\cdots\varw_\arityr$ in $\queryphi$ we
  include in $\relR^\DB$ the tuple
  $(\iotasubij(\varw_1),\ldots,\iotasubij(\varw_\arityr))$
  \begin{itemize}
   \item
    for all $\indi,\indj\in[\dimn]$ such that $\vecv_\indj=1$, 
    if $\sgpsi=\sgpsiy$,
   \item
    for all $\indi,\indj\in[\dimn]$ such that
    $\matM_{\indi,\indj}=1$, if $\sgpsi=\sgpsixy$, and
   \item
    for all $\indi,\indj\in[\dimn]$, if
    $\sgpsi\notin\set{\sgpsixy,\sgpsiy}$.
  \end{itemize}
  Note that $\iotasubij$ is a homomorphism 
  from $\phi$ to $\DB$
  if and only if
  $\matM_{\indi,\indj}=1$ and $\vecv_\indj=1$.
  Recall that 
  $\eval{\queryphi}{\DB}=
  \Setc{\big(\homh(x),\homh(z_1),\ldots,\homh(z_{\ell'})\big)}{\text{$\homh$ is a homomorphism from $\phi$ to $\DB$}}$.
  Because $\queryphi$ is
  self-join free,
  every homomorphism from $\queryphi$ to $\DB$ agrees
  with some $\iotasubij$.  Therefore,
  $\eval{\queryphi}{\DB}$ is the set of all tuples 
  $(\verta_\indi,\vertc_1,\ldots,\vertc_\ellalt)$ for which there
  exists an index $\indj$ such that $\iotasubij$ is
  a homomorphism from $\queryphi$ to $\DB$.
  Hence,
  $(\verta_\indi,\vertc_1,\ldots,\vertc_\ellalt)\in\eval{\queryphi}{\DB}$
  if and only if $(\matM\vecv)_i=1$. 
  As $\setsize{\eval{\queryphi}{\DB}}\leq\dimn$, we can 
  enumerate the entire query result $\eval{\queryphi}{\DB}$ in time
  $\bigoh(\dimn\delaytime)=\bigoh(\dimn^{2-\smalleps})$, and from this query result we can 
  easily compute the vector $\matM\vecv$. All this is done within time $\bigoh(\dimn^{2-\smalleps})$.
  
  When a vector $\vecv^{\,\indt}$ arrives in the dynamic phase of \OMv, we
  update $\DB(\queryphi,\matM,\vecv^{\,\indt-1})$ to
  $\DB(\queryphi,\matM,\vecv^{\,\indt})$ using at most $\dimn$
  insertions or deletions of tuples in the relation of the atom
  $\sgpsiy$.  As this can be done in time
  $\dimn\updatetime=\bigoh(n^{2-\smalleps})$, we can compute
  $\matM\vecv^{\,\indt}$ in overall time
  $\bigoh(\dimn^{2-\smalleps})$ and hence solve \OMv in time
  $\bigoh(\dimn^{3-\smalleps})$.
\end{proof}

\makeatletter{}%
\newcommand{\pointerfarben}{%
 \tikzstyle{pointer}=[draw,->,thick,double]
 \tikzstyle{pointerx}=[draw,->,thick,double]
 \tikzstyle{pointery}=[draw,->,thick,double,color=red]
 \tikzstyle{pointeryy}=[draw,->,thick,double,color=blue]
 \tikzstyle{pointerz}=[draw,->,thick,double,color=violet]
 \tikzstyle{pointerzz}=[draw,->,thick,double,color=brown]
}

\section{Upper Bound}\label{sec:upperbound}

This section is devoted to the proof of Theorem~\ref{thm:upperbound}.
First of all, note that if we can prove the theorem for \emph{connected} \qhier queries, then the result easily follows
also for non-connected \qhier queries: 
If $\queryphi_1(\ov{x}_1),\ldots,\queryphi_j(\ov{x}_j)$ are the
connected components of a \qhier query
$\queryphi(\ov{x}_1,\ldots,\ov{x}_j)$ for (possibly empty) tuples 
$\ov{x}_1,\ldots,\ov{x}_j$ of variables, 
then during the preprocessing phase we build the data structures for
all the $\queryphi_i$, and 
when the database is updated, we update all these data structures.
Note that $\queryphi(\DB)=\queryphi_1(\DB)\times\cdots\times\queryphi_j(\DB)$. 
Thus, the $\COUNT$ routine for $\queryphi$ can perform a $\COUNT$ for each $\queryphi_i$ and output the number 
$\setsize{\queryphi(\DB)}= \prod_{i=1}^j \setsize{\queryphi_i(\DB)}$.
Accordingly, the $\ENUMERATE$ routine can easily be obtained by a nested loop through the 
$\ENUMERATE$ routines for all the $\queryphi_i$'s
For the remainder of this section we assume w.l.o.g.\ that
$\queryphi(x_1,\ldots,x_k)$ is a \emph{connected} \qhier conjunctive
query, $\Vars(\queryphi)=\set{x_1,\ldots,x_m}$ with
$0\leq k\leq m$, and $\queryphi$ is of the form
\begin{equation}
 \phi \quad = \quad
 \exists x_{k+1}\cdots\exists x_m \body{ \qatom_1 \und \cdots \und \qatom_\ell}\,.
\end{equation}
From Lemma~\ref{lem:querytree} we know that $\queryphi$ has a
\Querytree. We use the lemma's algorithm to construct in time
$\poly(\queryphi)$ a \Querytree $\querytree$ of
$\queryphi$. For the remainder of this section, we simply write $T$ to
denote $\querytree$.    
Recall that the vertex set $V$ of $T$ is the set of variables in $\queryphi$,
i.e., $V=\set{x_1,\ldots,x_m}$. We will write $\vroot$ to denote the
root node of $T$.

\subsection{Further notation}

The following notation will be convenient for describing and analysing
our algorithm.
For a node $\varv$ of $T$, we write $\pa[\varv]$ to denote the set of all
nodes of $T$ that occur in the path from $\vroot$ to $\varv$ in $T$
(including $\varv$),
and we let $\pa[\varv) := \pa[\varv] \setminus \{v\}$. 
$N(\varv) := \setc{\varu}{(\varv,\varu)\in E(T)}$ is the set of children of $\varv$ in $T$.
A node $\varv$ of $T$ \emph{represents} an atomic query
$\qatom$ iff $\pa[v]=\Vars(\qatom)$, i.e., 
the variables in $\pa[\varv]$ are exactly the variables in $\qatom$. 
For each $\varv \in V$, we write $\atm{\varv}$ for the set of all atoms
$\qatom_j$ of $\query$ that are represented by $\varv$. 
Note that $\atm{v}\subseteq\atoms(v)$.

An \emph{assignment} %
is a partial mapping from $\Var$ to $\Dom$. %
As usual, we write $\dom{\assign}$ for the domain of $\assign$. For a
set $S\subseteq \Var$, by  $\restrict{\assign}{S}$ we denote the
restriction of $\assign$ to $\dom{\assign}\cap S$.
For  $x\in \Var$ and
$a\in\Dom$
we write $\extend{\assign}{x}{a}$ for the assignment $\assign'$ with domain
$\dom{\assign}\cup\set{x}$, where $\assign'(x)=a$ and
$\assign'(y)=\assign(y)$ for all
$y\notin \dom{\assign}\setminus\set{x}$.
An assignment $\beta$ is called an \emph{expansion} of $\assign$
(for short: $\beta\supseteq\assign$) if
$\dom{\beta}\supseteq\dom{\assign}$ and $\restrict{\beta}{\dom{\assign}}=\assign$.
The \emph{empty assignment} $\emptyassign$ is the assignment with empty domain.  
For pairwise distinct variables $\varv_1,\ldots,\varv_d$ and constants
$a_1,\ldots,a_d\in\Dom$ we write
$\Assign{\varv_1,\ldots,\varv_{d}}{a_1,\ldots,a_d}$ to denote the
assignment $\assign$ with $\dom{\assign}=\set{\varv_1,\ldots,\varv_d}$
and $\assign(\varv_i)=a_i$ for all $i\in[d]$.

\subsection{The data structure} \label{sec:data_structure}

\newcommand{\itesingle}[1]{\ensuremath{\texttt{\upshape\large[}#1\texttt{\upshape\large]}}}
\renewcommand{\itesingle}[1]{\ensuremath{\bm{[}#1\bm{]}}}

We now describe the data structure $\DS$ that 
will be built by the $\PREPROCESS$ routine and maintained while
executing the $\UPDATE$ routine.
Our data structure for a given database $\DB$ represents so-called \emph{items}.
Each item is determined by a variable $v\in V$, an assignment
$\assign:\pa[v)\to \Dom$, and a constant $a\in\Dom$;
we will write 
$\ite{\varv}{\assign}{a}$ to denote this item. 
For an item $i=\ite{\varv}{\assign}{a}$
we write 
$\varitem{i}$, $\assitem{i}$, and $\constitem{i}$ to denote the item's
variable $v$, assignment $\alpha$, and constant $a$.
Moreover, for every item $i=\ite{\varv}{\assign}{a}$ and every child $\varu$ of $\varv$ in
$T$ there is a doubly linked list $\llist{i}{\varu}$ (the \emph{$\varu$-list} of $i$) which contains items
of the form $\ite{\varu}{\extend{\assign}{\varv}{a}}{b}$ and we have one pointer $\childitem{i}{\varu}$ that points
from $i$ to the first element in $\llist{i}{\varu}$.
It is important to note that not every item of the form
$\ite{\varu}{\extend{\assign}{\varv}{a}}{b}$ that is present in our
data structure will be contained in the corresponding list
$\llist{i}{\varu}$, but for those items that are contained we store 
two pointers
$\nextlistitem{i}$ and $\prevlistitem{i}$
to navigate in $\llist{i}{\varu}$.
The \emph{parent item} of an item of the form
$\ite{\varu}{\extend{\assign}{\varv}{a}}{b}$, where $u$ is a child of
$v$ in $T$, is defined to be the item
$\ite{\varv}{\assign}{a}$.

Let us now state which items are actually contained in our data structure.
\newcommand{\existenceConditionText}{%
An item $\ite{\varv}{\assign}{a}$ is present in our data structure
if and only if there is an atom $\qatom\in\atoms(\varv)$
such that there is an expansion
$\beta\supseteq\extend{\assign}{\varv}{a}$ with
$\dom{\beta}=\Vars(\qatom)$ and $\satisfy{\DB}{\beta}{\qatom}$.\xspace}
\existenceConditionText{}
It follows that every fact $\relR(a_1,\ldots,a_r)$ in the database
gives rise to a constant number of items and that the overall number of items in our
data structure is therefore linear in the size of the database.
The definition also ensures that whenever an item is present in our
data structure, then so is its parent item.

Let us now specify which of the present items are actually contained
in the corresponding list $\llist{i}{\varu}$. We will need
the following definition.

An item $\ite{\varv}{\assign}{a}$ is \emph{fit}
if and only if
there is an 
expansion
$\beta\supseteq\extend{\assign}{\varv}{a}$ such that
$\satisfy{\DB}{\beta}{\bigwedge_{\qatom\in\atoms(\varv)}\qatom}$.
The doubly linked list $\llist{i}{\varu}$ contains precisely those
items that are fit.
Being fit is a necessary requirement for participating in
the query result and this is why we exclude unfit items from the lists.
Note that whenever a tuple is inserted into or deleted from the database,
this affects the ``fit''-status of only a constant number of
items. Furthermore, provided we have 
constant-time 
access to the items, we can
update their status and include or exclude them from the
corresponding lists in constant time. 
In addition to the items, our data structure also has a designated 
$\start$ pointer that points to (the first element of) a doubly linked
list \start-list
$\startlist$,
which consists of all fit items of the form $\ite{\vroot}{\emptyset}{b}$.

In order to count the number of output tuples, we store for every item $i$
its \emph{weight} $\fitcount{i}$. The weight is used to measure the number of
tuples in the query result that extend the item's partial assignment
and is defined as follows.
For an item $i=\ite{\varv}{\assign}{a}$ we let
\begin{align*}
  \extensionsetof{i} &\ \ \deff \ \ \setc{\,\beta\supseteq\extend{\assign}{\varv}{a}}{\dom\beta =
    \textstyle\bigcup_{\qatom\in\atoms(\varv)}\Vars(\qatom);\quad
  \satisfy{\DB}{\beta}{\textstyle\bigwedge_{\qatom\in\atoms(\varv)}\qatom}\,},
  \\
  \fitcount{i} &\ \ \deff \ \ \Setsize{\extensionsetof{i}}.
\end{align*}
By definition we have that $i$ is fit if and only if
$\fitcount{i}>0$, and we will use the weights to quickly determine
whether an 
item becomes fit after an insertion of a tuple (or unfit
after a deletion).
In order to efficiently 
update these numbers, we store for every
list $\llist{i}{\varu}$ in our data structure as well for the
\start-list the sum of the weights of their elements
\begin{align}
  \label{eq:9}
 \listcount{i}{u} &\ \ \deff \ \ \textstyle\sum_{i'\in\llist{i}{\varu}} \fitcount{i'}
                    \qquad\text{ and }\qquad
 \startcount \ \ \deff \ \  \textstyle\sum_{i\in\startlist} \fitcount{i}.
\end{align}
The weights will also be used to determine the size of the query
result. For example, suppose that $\query$ is quantifier-free. Then it follows from the definition of that $\eval{\query}{\DB}$
is the disjoint union of the sets $\extensionsetof{i}$ for all
$i\in\startlist$ (where every tuple in $\eval{\query}{\DB}$ is viewed
as an assignment \mbox{$\beta:\free(\query)\to\Dom$).} Therefore,
$\setsize{\eval{\query}{\DB}}=\startcount$ and we can immediately
answer a $\COUNT$ request by reporting the stored value $\startcount$.
Now suppose that $\varphi$ is Boolean. Then we respond to an $\ANSWER$
request by reporting whether $\startcount>0$.
The remaining case, when $\query$ contains quantified and free
variables, is similar and will be handled in Section~\ref{sec:upperbound_counting}.
To illustrate the overall shape of our data structure,
let us consider the following example.
 
\ifthenelse{\isundefined{\ARTICLEFORMAT}}{
\makeatletter{}%
\begin{figure}[pt]
\begin{tikzpicture}[scale=0.85] 
 \node at (-4,2.2) {{\Querytree }$T$:};
 \node[draw,circle,thick] (x) at (0,2) {$\varx\vphantom{y'}$};
   \node at (1.4,2) {\scriptsize$\atm{x}=\emptyset$};
 \node[draw,circle,thick] (y) at (-1,1) {$\vary\vphantom{y'}$};
   \node at (-2.8,1) {\scriptsize$\atm{y}=\set{Exy}$};
 \node[draw,circle,thick] (y') at (1,1) {$\vary'$};
   \node at (2.9,1) {\scriptsize$\atm{y'}=\set{Exy'}$};
 \node[draw,circle,thick] (z) at (-2,0) {$\varz\vphantom{y'}$};
   \node at (-2.8,-0.65) {\scriptsize$\atm{z}=\set{Rxyz\,,\, Syxz}$};
 \node[draw,circle,thick] (z') at (0,0) {$\varz'\vphantom{y'}$};
   \node at (1.9,-0.2) {\scriptsize$\atm{z'}=\set{Rxyz'}$};
 \draw[->,thick] (x) -- (y);
 \draw[->,thick] (x) -- (y');
 \draw[->,thick] (y) -- (z);
 \draw[->,thick] (y) -- (z');

\begin{scope}[yshift=-1.5cm,xshift=-3cm]
  \node at (-1.5,-1) {$\DS$:};
 \tikzstyle{items1}=[fill=white,draw,thick,rectangle]
 \tikzstyle{items2}=[fill=black,thick,rectangle]
 \tikzstyle{list}=[thick]
 \tikzstyle{pointer}=[->,thick,double]
 \pointerfarben
 
 \node (start) at (-0.25,0) {$\start$};
 
 \newcommand{\xvalueformula}[1]{#1 * 0.7 + 1}
 \foreach \x in {0,1,2,3,4,5,6} {
  \pgfmathsetmacro{\xvalue}{\xvalueformula{\x}}
  \node[items1] (x_\x) at (\xvalue,0) {};
 }
 \pgfmathsetmacro{\xvalue}{\xvalueformula{7}}
 \node (x_7) at (\xvalue,0) {};
 
 \draw[list] (x_0) -- (x_1);
 \draw[list] (x_1) -- (x_2);
 \draw[list] (x_2) -- (x_3);
 \draw[list] (x_3) -- (x_4);
 \draw[list] (x_4) -- (x_5);
 \draw[list] (x_5) -- (x_6);
 \draw[list,dashed] (x_6) -- (x_7);
 \pgfmathsetmacro{\xvalue}{\xvalueformula{4}}
 \node[items2,label={above:$a$}] (sel_x) at (\xvalue,0) {};
 
 \renewcommand{\xvalueformula}[1]{#1 * 0.7 + 4}
 \foreach \x in {0,1,2,3,4} {
  \pgfmathsetmacro{\xvalue}{\xvalueformula{\x}}
  \node[items1] (y'_\x) at (\xvalue,-1) {};
 }
 \pgfmathsetmacro{\xvalue}{\xvalueformula{5}}
 \node (y'_5) at (\xvalue,-1) {};
 \draw[list] (y'_0) -- (y'_1);
 \draw[list] (y'_1) -- (y'_2);
 \draw[list,dashed] (y'_2) -- (y'_3);
 \draw[list] (y'_3) -- (y'_4);
 \draw[list,dashed] (y'_4) -- (y'_5);
 \pgfmathsetmacro{\xvalue}{\xvalueformula{2}}
 \node[items2,label={above:$b$}] (sel_y') at (\xvalue,-1) {};
 \pgfmathsetmacro{\xvalue}{\xvalueformula{3}}
 \node[items2,label={above:$d$}] (sel_y') at (\xvalue,-1) {};
 
 \renewcommand{\xvalueformula}[1]{(-1) * #1 * 0.7 + 3}
 \foreach \x in {3,2,1,0} {
  \pgfmathsetmacro{\xvalue}{\xvalueformula{\x}}
  \node[items1] (y_\x) at (\xvalue,-1) {};
 }
 \pgfmathsetmacro{\xvalue}{\xvalueformula{4}}
 \node (y_4) at (\xvalue,-1) {};
 \draw[list] (y_0) -- (y_1);
 \draw[list] (y_1) -- (y_2);
 \draw[list] (y_3) -- (y_2);
 \draw[list,dashed] (y_3) -- (y_4);
 \pgfmathsetmacro{\xvalue}{\xvalueformula{2}}
 \node[items2,label={above:$b$}] (sel_y) at (\xvalue,-1) {};
 
 \renewcommand{\xvalueformula}[1]{#1 * 0.7 + 3}
 \foreach \x in {0,1,2,3,4} {
  \pgfmathsetmacro{\xvalue}{\xvalueformula{\x}}
  \node[items1] (z'_\x) at (\xvalue,-2) {};
 }
 \pgfmathsetmacro{\xvalue}{\xvalueformula{5}}
 \node (z'_5) at (\xvalue,-2) {};
 \draw[list] (z'_0) -- (z'_1);
 \draw[list] (z'_1) -- (z'_2);
 \draw[list,dashed] (z'_2) -- (z'_3);
 \draw[list] (z'_3) -- (z'_4);
 \draw[list,dashed] (z'_3) -- (z'_4);
 \pgfmathsetmacro{\xvalue}{\xvalueformula{2}}
 \node[items2,label={above:$c$}] (sel_z') at (\xvalue,-2) {};
 \pgfmathsetmacro{\xvalue}{\xvalueformula{3}}
 \node[items2,label={above:$e$}] (sel_z') at (\xvalue,-2) {};
 
 \renewcommand{\xvalueformula}[1]{(-1) * #1 * 0.7 + 2}
 \foreach \x in {3,2,1,0} {
  \pgfmathsetmacro{\xvalue}{\xvalueformula{\x}}
  \node[items1] (z_\x) at (\xvalue,-2) {};
 }
 \pgfmathsetmacro{\xvalue}{\xvalueformula{4}}
 \node (z_4) at (\xvalue,-2) {};
 \draw[list] (z_0) -- (z_1);
 \draw[list] (z_1) -- (z_2);
 \draw[list] (z_3) -- (z_2);
 \draw[list,dashed] (z_3) -- (z_4);
 \pgfmathsetmacro{\xvalue}{\xvalueformula{2}}
 \node[items2,label={above:$c$}] (sel_z) at (\xvalue,-2) {};
 
 \draw[pointerx] (start) -- (x_0); 
 \draw[pointery] (sel_x) -- (y_0) node [midway, left, draw=none] {\
   \scriptsize $\vary$};
 \draw[pointeryy] (sel_x) -- (y'_0) node [midway, right, draw=none] {
   \scriptsize $\vary'$};
 \draw[pointerz] (sel_y) -- (z_0) node [midway, left, draw=none] {
   \scriptsize $\varz$};
 \draw[pointerzz] (sel_y) -- (z'_0) node [midway, right, draw=none] {\ 
   \scriptsize $\varz'$}; 
\end{scope}
\end{tikzpicture}
\caption{A $\Querytree$ $T$ along with the atoms represented by the
  tree's nodes, and an illustration of the data structure
  $\DS$ for Example~\ref{upperbound:ex1}.}\label{upperbound:ex1:fig1}
\end{figure}

}{
\makeatletter{}%
\begin{figure}[pt]
  \begin{tikzpicture}   
 \node[draw,circle,thick,inner sep = 0mm, minimum size =
 6mm,label=right:{\scriptsize\begin{tabular}{l}$\atm{x}=\emptyset$\\ $\atoms(x)=$\set{R\varx\vary\varz, R\varx\vary\varz', E\varx\vary, E\varx\vary', S\varx\vary\varz}\end{tabular}}] (x) at (0,2) {$\varx\vphantom{y'}$};
 \node[draw,circle,thick,inner sep = 0mm, minimum size = 6mm,label=left:{\scriptsize\begin{tabular}{r}$\atm{y}=\set{E\varx\vary}$\\ $\atoms(y)=$\set{R\varx\vary\varz, R\varx\vary\varz', E\varx\vary, S\varx\vary\varz}\end{tabular}}] (y) at (-1,1) {$\vary\vphantom{y'}$};
 \node[draw,circle,thick,inner sep = 0mm, minimum size = 6mm,label=right:{\scriptsize$\atm{y'}=\atoms(y')=\set{Exy'}$}] (y') at (1,1) {$\vary'$};
 \node[draw,circle,thick,inner sep = 0mm, minimum size = 6mm, label=left:{\scriptsize$\atm{z}=\atoms(z)=\set{Rxyz\,,\, Sxyz}$}] (z) at (-2,0) {$\varz\vphantom{y'}$};
 \node[draw,circle,thick,inner sep = 0mm, minimum size = 6mm, label=right:{\scriptsize$\atm{z'}=\atoms(z')=\set{Rxyz'}$}] (z') at (0,0) {$\varz'\vphantom{y'}$};
 \draw[->,thick] (x) -- (y);
 \draw[->,thick] (x) -- (y');
 \draw[->,thick] (y) -- (z);
 \draw[->,thick] (y) -- (z');
 \node at (-5,2) {\scriptsize $\query= 
\big(\,R\varx\vary\varz \uund R\varx\vary\varz' \uund E\varx\vary \uund E\varx\vary' \uund S\varx\vary\varz\,\big)
 $};
\end{tikzpicture}
\centering
\caption{A $\Querytree$ $T$ for the query in Example~\ref{upperbound:ex1} along with the atoms represented by the
  tree's nodes.}\label{upperbound:ex1:fig1}
\end{figure}
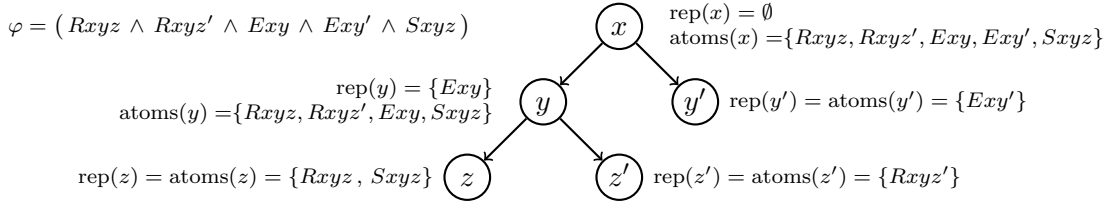

}
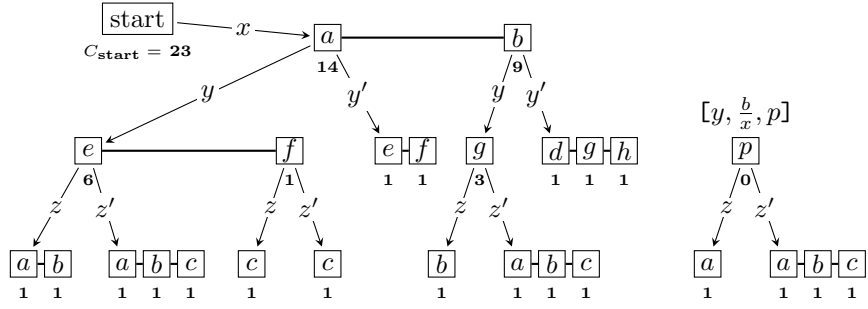
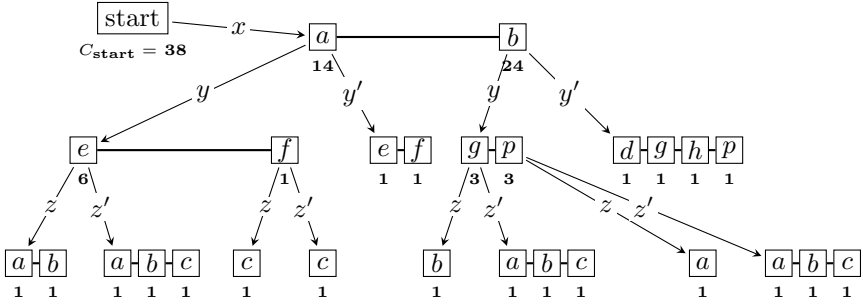
\begin{figure}[pt]
\centering
\renewcommand{\DBzero}[1]{#1}
\renewcommand{\DBone}[1]{}
\subfigure[Data structure $\DS_0$ for the database $\DBstart$.]{%
\makeatletter{}%

\begin{tikzpicture}
  \begin{scope}
 \tikzstyle{items1}=[fill=white,draw,thick,rectangle]
 \tikzstyle{items2}=[fill=black,thick,rectangle]
 \tikzstyle{list}=[-,thick]
 \tikzstyle{pointer}=[->,>=stealth,shorten >=0.5mm,shorten <=0.5mm]

 \pgfmathsetmacro{\boxsize}{0.35cm}
 \pgfmathsetmacro{\listdist}{0.1+0.35}
 \pgfmathsetmacro{\ydist}{1.5}
 
 \tikzstyle{itembox}=[fill=white,draw=black,rectangle,inner sep = 0, minimum size = \boxsize]
 \newcommand{\nrbox}[1]{\node[itembox,label=below:{\tiny\textbf{#1}}]}

 \nrbox{1} (a1) at (0.5 + 0*\listdist,0) {\small $a$};
 \nrbox{1} (a2) at (0.5 + 1*\listdist,0) {\small $b$};

 \draw[list] (a1) -- (a2);

 \nrbox{1} (bb1) at (1.8,0) {\small $a$};
 \nrbox{1} (bb2) at (1.8 + \listdist,0) {\small $b$};
 \nrbox{1} (bb3) at (1.8 + 2*\listdist,0) {\small $c$};

 \draw[list] (bb1) -- (bb2);
 \draw[list] (bb2) -- (bb3);

 \nrbox{1} (b2) at (3.5,0) {\small $c$};
 \nrbox{1} (b3) at (4.5,0) {\small $c$};
 \nrbox{1} (b4) at (6,0) {\small $b$};
 
 \nrbox{1} (e1) at (7 + 0*\listdist,0) {\small $a$};
 \nrbox{1} (e2) at (7 + 1*\listdist,0) {\small $b$};
 \nrbox{1} (e3) at (7 + 2*\listdist,0) {\small $c$};

 \draw[list] (e1) -- (e2);
 \draw[list] (e2) -- (e3);

 \nrbox{1} (f1) at (9.5,0) {\small $a$};
 
 \nrbox{1} (g1) at (10.5 + 0*\listdist,0) {\small $a$};
 \nrbox{1} (g2) at (10.5 + 1*\listdist,0) {\small $b$};
 \nrbox{1} (g3) at (10.5 + 2*\listdist,0) {\small $c$};

 \draw[list] (g1) -- (g2);
 \draw[list] (g2) -- (g3);

 \nrbox{6} (A1) at (0 + 3*\listdist,\ydist) {\small $e$};
 \nrbox{1} (A2) at (4,\ydist) {\small $f$};
  \draw[list] (A1) -- (A2); 

 \nrbox{1} (B1) at (5+0.3,\ydist) {\small $e$};
 \nrbox{1} (B2) at (5+0.3+\listdist,\ydist) {\small $f$};
 \draw[list] (B1) -- (B2);
 
 \nrbox{3} (C1) at (6.5,\ydist) {\small $g$};

 \nrbox{1} (E1) at (7.5 \DBone{+ 1} + 0*\listdist,\ydist) {\small $d$};
 \nrbox{1} (E2) at (7.5 \DBone{+ 1} + 1*\listdist,\ydist) {\small $g$};
 \nrbox{1} (E3) at (7.5 \DBone{+ 1} + 2*\listdist,\ydist) {\small $h$};

 \draw[list] (E1) -- (E2);
 \draw[list] (E2) -- (E3);

\node[rectangle,draw=white] at (12.5,2*\ydist) {};
 
\DBzero{
 \nrbox{0} (F1) at (10,\ydist) {\small $p$};
 \node at (10,\ydist+0.5) {\small \ite{y}{\Assign{x}{b}}{p}};

}%

\DBone{
 \nrbox{3} (F1) at (6.5+\listdist,\ydist) {\small $p$};
 \draw[list] (C1) -- (F1);
 
 \nrbox{1} (G1) at (7.5 \DBone{+ 1} + 3*\listdist,\ydist) {\small $p$};
 \draw[list] (E3) -- (G1);
}%

 \node[fill=white,draw=black,rectangle,inner sep =
 1mm,label=below:{\tiny\textbf{$\startcount=$ \DBone{38}\DBzero{23}}}] (start) at (2,0.25+2*\ydist) {\small start};
 \nrbox{14} (AA1) at (4.5,2*\ydist) {\small $a$};
 \nrbox{\DBzero{9}\DBone{24}} (AA2) at (7,2*\ydist) {\small $b$};
 \draw[list] (AA1) -- (AA2);

 \newcommand{\varlabel}[1]{node [midway,fill=white,inner sep = 0.5mm] {\small $#1$}}
 
 \draw[pointer] (start) -- (AA1) \varlabel{x};
 
 \draw[pointer] (AA1) -- (A1) \varlabel{y};
 \draw[pointer] (AA1) -- (B1) \varlabel{y'};
 \draw[pointer] (AA2) -- (C1) \varlabel{y};
 \draw[pointer] (AA2) -- (E1) \varlabel{y'};

 \draw[pointer] (A1) -- (a1) \varlabel{z};
 \draw[pointer] (A1) -- (bb1) \varlabel{z'};
 \draw[pointer] (A2) -- (b2) \varlabel{z};
 \draw[pointer] (A2) -- (b3) \varlabel{z'};
 
 \draw[pointer] (C1) -- (b4) \varlabel{z};
 \draw[pointer] (C1) -- (e1) \varlabel{z'};
 
 \draw[pointer] (F1) -- (f1) \varlabel{z};
 \draw[pointer] (F1) -- (g1) \varlabel{z'};

 \end{scope}
\end{tikzpicture}

\label{fig:DSzero}}
\renewcommand{\DBzero}[1]{}
\renewcommand{\DBone}[1]{#1}
\subfigure[Data structure $\DS_1$
   after
  inserting $E(b,p)$ into $\DBstart$.]{%
\makeatletter{}%

\begin{tikzpicture}
  \begin{scope}
 \tikzstyle{items1}=[fill=white,draw,thick,rectangle]
 \tikzstyle{items2}=[fill=black,thick,rectangle]
 \tikzstyle{list}=[-,thick]
 \tikzstyle{pointer}=[->,>=stealth,shorten >=0.5mm,shorten <=0.5mm]

 \pgfmathsetmacro{\boxsize}{0.35cm}
 \pgfmathsetmacro{\listdist}{0.1+0.35}
 \pgfmathsetmacro{\ydist}{1.5}
 
 \tikzstyle{itembox}=[fill=white,draw=black,rectangle,inner sep = 0, minimum size = \boxsize]
 \newcommand{\nrbox}[1]{\node[itembox,label=below:{\tiny\textbf{#1}}]}

 \nrbox{1} (a1) at (0.5 + 0*\listdist,0) {\small $a$};
 \nrbox{1} (a2) at (0.5 + 1*\listdist,0) {\small $b$};

 \draw[list] (a1) -- (a2);

 \nrbox{1} (bb1) at (1.8,0) {\small $a$};
 \nrbox{1} (bb2) at (1.8 + \listdist,0) {\small $b$};
 \nrbox{1} (bb3) at (1.8 + 2*\listdist,0) {\small $c$};

 \draw[list] (bb1) -- (bb2);
 \draw[list] (bb2) -- (bb3);

 \nrbox{1} (b2) at (3.5,0) {\small $c$};
 \nrbox{1} (b3) at (4.5,0) {\small $c$};
 \nrbox{1} (b4) at (6,0) {\small $b$};
 
 \nrbox{1} (e1) at (7 + 0*\listdist,0) {\small $a$};
 \nrbox{1} (e2) at (7 + 1*\listdist,0) {\small $b$};
 \nrbox{1} (e3) at (7 + 2*\listdist,0) {\small $c$};

 \draw[list] (e1) -- (e2);
 \draw[list] (e2) -- (e3);

 \nrbox{1} (f1) at (9.5,0) {\small $a$};
 
 \nrbox{1} (g1) at (10.5 + 0*\listdist,0) {\small $a$};
 \nrbox{1} (g2) at (10.5 + 1*\listdist,0) {\small $b$};
 \nrbox{1} (g3) at (10.5 + 2*\listdist,0) {\small $c$};

 \draw[list] (g1) -- (g2);
 \draw[list] (g2) -- (g3);

 \nrbox{6} (A1) at (0 + 3*\listdist,\ydist) {\small $e$};
 \nrbox{1} (A2) at (4,\ydist) {\small $f$};
  \draw[list] (A1) -- (A2); 

 \nrbox{1} (B1) at (5+0.3,\ydist) {\small $e$};
 \nrbox{1} (B2) at (5+0.3+\listdist,\ydist) {\small $f$};
 \draw[list] (B1) -- (B2);
 
 \nrbox{3} (C1) at (6.5,\ydist) {\small $g$};

 \nrbox{1} (E1) at (7.5 \DBone{+ 1} + 0*\listdist,\ydist) {\small $d$};
 \nrbox{1} (E2) at (7.5 \DBone{+ 1} + 1*\listdist,\ydist) {\small $g$};
 \nrbox{1} (E3) at (7.5 \DBone{+ 1} + 2*\listdist,\ydist) {\small $h$};

 \draw[list] (E1) -- (E2);
 \draw[list] (E2) -- (E3);

\node[rectangle,draw=white] at (12.5,2*\ydist) {};
 
\DBzero{
 \nrbox{0} (F1) at (10,\ydist) {\small $p$};
 \node at (10,\ydist+0.5) {\small \ite{y}{\Assign{x}{b}}{p}};

}%

\DBone{
 \nrbox{3} (F1) at (6.5+\listdist,\ydist) {\small $p$};
 \draw[list] (C1) -- (F1);
 
 \nrbox{1} (G1) at (7.5 \DBone{+ 1} + 3*\listdist,\ydist) {\small $p$};
 \draw[list] (E3) -- (G1);
}%

 \node[fill=white,draw=black,rectangle,inner sep =
 1mm,label=below:{\tiny\textbf{$\startcount=$ \DBone{38}\DBzero{23}}}] (start) at (2,0.25+2*\ydist) {\small start};
 \nrbox{14} (AA1) at (4.5,2*\ydist) {\small $a$};
 \nrbox{\DBzero{9}\DBone{24}} (AA2) at (7,2*\ydist) {\small $b$};
 \draw[list] (AA1) -- (AA2);

 \newcommand{\varlabel}[1]{node [midway,fill=white,inner sep = 0.5mm] {\small $#1$}}
 
 \draw[pointer] (start) -- (AA1) \varlabel{x};
 
 \draw[pointer] (AA1) -- (A1) \varlabel{y};
 \draw[pointer] (AA1) -- (B1) \varlabel{y'};
 \draw[pointer] (AA2) -- (C1) \varlabel{y};
 \draw[pointer] (AA2) -- (E1) \varlabel{y'};

 \draw[pointer] (A1) -- (a1) \varlabel{z};
 \draw[pointer] (A1) -- (bb1) \varlabel{z'};
 \draw[pointer] (A2) -- (b2) \varlabel{z};
 \draw[pointer] (A2) -- (b3) \varlabel{z'};
 
 \draw[pointer] (C1) -- (b4) \varlabel{z};
 \draw[pointer] (C1) -- (e1) \varlabel{z'};
 
 \draw[pointer] (F1) -- (f1) \varlabel{z};
 \draw[pointer] (F1) -- (g1) \varlabel{z'};

 \end{scope}
\end{tikzpicture}

\label{fig:DSone}}
 \caption{Data structure
   for the database in
   Example~\ref{upperbound:ex1}. Some unfit items are omitted.}
  \label{fig:datastructure}
\end{figure}

\begin{example}\label{upperbound:ex1}
  Consider the query \,$\query(x,y,z,y',z')\deff$
  \[
\big(\;R\varx\vary\varz \uund R\varx\vary\varz' \uund E\varx\vary \uund E\varx\vary' \uund S\varx\vary\varz\;\big)
  \]
 and 
 the database $\DBstart$ with
 $E^\DBstart=\set{
   (a,e),\,\allowbreak
   (a,f),\,\allowbreak
   (b,d),\,\allowbreak
   (b,g),\,\allowbreak
   (b,h)
 }$, 
 $S^\DBstart=\set{
   (a,e,a),\,\allowbreak 
   (a,e,b),\,\allowbreak 
   (a,f,c),\,\allowbreak 
   (b,g,b),\,\allowbreak 
   (b,p,a)
 }$, and 
 $R^\DBstart=S^\DBstart\,\cup\,\set{
   (a,e,c),\,\allowbreak
   (b,g,a),\,\allowbreak
   (b,g,c),\,\allowbreak
   (b,p,b),\,\allowbreak
   (b,p,c)
 }$. 
  Figure~\ref{upperbound:ex1:fig1} depicts a
  \Querytree $T$ for $\query$. 
  The data structure $\DS_0$ that represents the database $\DBstart$
  is shown in Figure~\ref{fig:DSzero}.
  Every box represents an item $i=\ite{\varv}{\alpha}{a}$ and contains
  the item's constant $a$. The number below each box is the weight
  $\fitcount{i}$ of the item.
  The 
  arrows labelled with $\varu$ represent the pointer from an item $i$
  to the first element of its $\varu$-list
  $\llist{i}{\varu}$. Horizontal lines between the items indicate
  pointers in the doubly linked
  lists $\llist{i}{\varu}$ and $\startlist$.
  There are seven further unfit items
  (omitted in the figure) that do
  not have connections to any other items:   $\ite{y}{\Assign{x}{b}}{d}$,
  $\ite{y}{\Assign{x}{b}}{h}$,
  $\ite{z}{\Assign{x,y}{a,e}}{c}$,
  $\ite{z}{\Assign{x,y}{b,g}}{a}$,
  $\ite{z}{\Assign{x,y}{b,g}}{c}$,
  $\ite{z}{\Assign{x,y}{b,p}}{b}$, and
  $\ite{z}{\Assign{x,y}{b,p}}{c}$.
 \end{example}

To enable quick access to the existing items, we store a number of 
arrays, where the elements of 
$\Dom$ are used to index the positions in an array.
For every $\varv\in V$ we store a $d$-ary array
$\larray{\varv}$, where $d := |\pa[\varv]|$.
We let $\varv_1,\ldots,\varv_d$ be the list of variables in
$\pa[\varv]$ in the order in which they are encountered in the path
from $T$'s root to the node $\varv$ (in particular, $\varv_1=\vroot$
and $\varv_d=\varv$).
For constants $a_1,\ldots,a_d\in\Dom$, the entry
$\larray{\varv}[a_1,\ldots,a_d]$ 
represents the item 
$\ite{\varv}{\Assign{\varv_1,\ldots,\varv_{d-1}}{a_1,\ldots,a_{d-1}}}{a_d}$,
if such an item exists. 
Otherwise, the entry $\larray{\varv}[a_1,\ldots,a_d]$ is initialised to
$\Null$.
This will enable us to check in constant time whether a particular item
exists, and if so, to also access this item within constant
time.
In the following, whenever an item 
$i=\ite{\varv}{\Assign{\varv_1,\ldots,\varv_{d-1}}{a_1,\ldots,a_{d-1}}}{a_d}$
is newly created, we tacitly let $\larray{\varv}[a_1,\ldots,a_d]\deff
i$, and we set $\larray{\varv}[a_1,\ldots,a_d]\deff
0$ if $i$ is deleted.
Note that while the support of all arrays (i.e., the non-$\Null$ entries) is
linear in the size of the current database $\DB$, a huge amount of
storage has to be reserved for these arrays.
In more practical settings one has to replace these arrays by more
space-efficient
data structures, such as suitable hash tables, that allow quick access to the
items.

\subsection{Enumeration}

\newcommand{\I}[1]{\multicolumn{1}{|c}{#1}}
\newcommand{\II}[1]{\multicolumn{1}{|c|}{#1}}
\begin{table}
  \centering
\begin{tabular}{|l||ccccccccccccccccccccccc|}
  \hline
  $\varx$  &{$a$}&{$a$}&{$a$}&{$a$}&{$a$}&{$a$}&{$a$}&{$a$}&{$a$}&{$a$}&{$a$}&{$a$}&{$a$}&{$a$}&\I{$b$}&{$b$}&{$b$}&{$b$}&{$b$}&{$b$}&{$b$}&{$b$}&{$b$}\\\hline
  $\vary$  &{$e$}&{$e$}&{$e$}&{$e$}&{$e$}&{$e$}&{$e$}&{$e$}&{$e$}&{$e$}&{$e$}&{$e$}&\I{$f$}&{$f$}&\I{$g$}&{$g$}&{$g$}&{$g$}&{$g$}&{$g$}&{$g$}&{$g$}&{$g$}\\\hline
  $\varz$  &{$a$}&{$a$}&{$a$}&{$a$}&{$a$}&{$a$}&\I{$b$}&{$b$}&{$b$}&{$b$}&{$b$}&{$b$}&\I{$c$}&{$c$}&\I{$b$}&{$b$}&{$b$}&{$b$}&{$b$}&{$b$}&{$b$}&{$b$}&{$b$}\\\hline
  $\varz'$ &{$a$}&{$a$}&\I{$b$}&{$b$}&\I{$c$}&{$c$}&\I{$a$}&{$a$}&\I{$b$}&{$b$}&\I{$c$}&{$c$}&\I{$c$}&{$c$}&\I{$a$}&{$a$}&{$a$}&\I{$b$}&{$b$}&{$b$}&\I{$c$}&{$c$}&{$c$}\\\hline
  $\vary'$ &{$e$}&\I{$f$}&\I{$e$}&\I{$f$}&\I{$e$}&\I{$f$}&\I{$e$}&\I{$f$}&\I{$e$}&\I{$f$}&\I{$e$}&\I{$f$}&\I{$e$}&\I{$f$}&\I{$d$}&\I{$g$}&\I{$h$}&\I{$d$}&\I{$g$}&\I{$h$}&\I{$d$}&\I{$g$}&\II{$h$}\\
  \hline
\end{tabular}
  \caption{Enumeration of $\eval{\query}{\DBstart}$ using data
    structure $\DS_0$ depicted in Figure~\ref{fig:DSzero}.}
  \label{tab:enumeration}
\end{table}

We now discuss how the data structure can be used to enumerate the
query result with constant delay.
For our Example~\ref{upperbound:ex1}, the $23$ result tuples of the
enumeration process are shown in Table~\ref{tab:enumeration}.
To enumerate the result of a non-Boolean conjunctive query
$\query(\varx_1,\ldots,\varx_k)$,
let $T'$ be the subtree of $T$ induced on $V'\deff\free(\query)=\set{\varx_1,\ldots,\varx_k}$.
Note that by the definition of a \Querytree, we know that $T'$ is
connected and contains $\vroot$.
For each node $\varv$ of $T'$, let us fix an (arbitrary) linear order
on the children of $\varv$ in $T'$. In our example query we have
$T'=T$ and we let
$\vary<\vary'$ and $\varz<\varz'$.
If the \start-list is empty, the $\ENUMERATE$ routine stops
immediately with output $\EOE$. Otherwise, we proceed as follows
to determine the first tuple in the query result. 
Let $i_{\vroot}$ be the
first item in the \start-list.  Inductively, for every $\varv\in V'$ for which $i_{\varv}$
has been chosen already, we choose $i_{\varu}$ for every child $\varu$
of $\varv$ in $T'$ 
by letting $i_{\varu}$ be the first item in the $\varu$-list of
$i_{\varv}$.
From the resulting items $\ITEMS\deff (i_{\varv})_{\varv\in V'}$
we obtain the first tuple $(a_1,\ldots,a_k)$ in the query result by
letting 
$a_\ell$ be the constant of item $i_{\varx_\ell}$ for each $\ell\in[k]$.
Thus, within time $\bigoh(k)$ we can output the first tuple
that belongs to the query result.

To jump from one output tuple to the next, we proceed as follows.
Assume that $\ITEMS=(i_{\varv})_{\varv\in V'}$ are the items which
determined the result tuple $(a_1,\ldots,a_k)$ that has just been output.
Let $y_1,\ldots,y_k$ be the list
of all nodes $V'$ of $T'$ in \emph{document order}, i.e., obtained by a
pre-order depth-first left-to-right traversal of $T'$ (in particular,
$y_1=\vroot$).
In Example~\ref{upperbound:ex1}, the vertices thus are ordered
$\varx,\vary,\varz,\varz',\vary'$.

Determine the maximum index $j\in[k]$ such that the
item $i_{\vary_j}$ is \emph{not} the last item of its doubly linked list.
If no such $j$ exists, stop with output $\EOE$. Otherwise, let $i'_{\vary_j}$
be the item indicated by $\nextlistitem{i_{\vary_j}}$. 
For every $\mu<j$ we let $i'_{\vary_\mu}\deff i_{\vary_\mu}$.
For $\mu = j{+}1,\ldots,k$ the item $i'_{\vary_\mu}$ is determined
inductively (along the document order) to be the first element of the
$\mu$-list of its parent.

From the resulting items $i'_{\vary_1},\ldots,i'_{\vary_k}$ 
we obtain the next tuple $(a'_1,\ldots,a'_k)$ in the query result by
letting 
$a'_\ell$ be the constant of item $i'_{\varx_\ell}$
for each $\ell\in[k]$.
Note that the delay between outputting two consecutive result tuples
is $\bigoh(k)$.  The result of the enumeration process for
Example~\ref{upperbound:ex1} is given in Table~\ref{tab:enumeration},
where the change of an item for a variable $\vary_\ell$ (i.e.
$i'_{\vary_\ell}\neq i_{\vary_\ell}$) for two consecutive tuples is
indicated by a separating line.

The pseudo-code for the described $\ENUMERATE$ routine is given in 
Algorithm~\ref{alg:enum}.
Our next goal is to prove the correctness of this algorithm. To this
end, we let
\begin{equation}
 \extensionsetfreeof{i}
 \ \ \deff \ \ 
 \setc{\,\restrict{\beta}{\free(\query)}}{\beta\in\extensionsetof{i}\,}
\end{equation}
and note that $(a_1,\ldots,a_k)\in\eval{\query}{\DB}$ if, and
only if, there is an $i\in\startlist$ such that the corresponding
$\beta\colon \free(\query)\to \Dom$ with $\beta(\varx_j)=a_j$ for all
$j\in[k]$ is contained in $\extensionsetfreeof{i}$.

\begin{algorithm}
\caption{Enumeration algorithm}
\label{alg:enum}
\begin{algorithmic}[1]
  \State \textbf{Input:} Data structure $\DS$, tree $T'$ with vertices
  $V'=\set{y_1,\ldots,y_k}$ (in document order) for query
  $\query(x_1,\ldots,x_k)$ (with $\set{x_1,\ldots,x_k}=V'$) 
  \State \textbf{Variables:} Item variables
  $i_{y_\ell}=\ite{y_\ell}{\alpha}{a_\ell}$ for $\ell\in[k]$.
 \State
 \If{$\startlist = \emptyset$}
 \State Halt and output the end-of-enumeration message $\EOE$.
 \EndIf
 \State Let $i_{y_1}$ be the first element of the start-list $\startlist$.
 \For{$\mu = 2$\ \textbf{to}\ $k$}
 \State $i_{y_\mu} \gets \textsc{Set}((i_{y_1},\ldots,i_{y_{\mu{-}1}}),\mu)$
 \EndFor
 \State \textsc{visit}$(i_{y_1},\ldots,i_{y_{k}})$
 \State
 \Function{Set}{$\ITEMS,\mu$} 
 \State \textbf{Input:} $\mu \in \set{2,\ldots,k}$ and a 
   list of items
 $\ITEMS=(i_{y_1},\ldots,i_{y_{\mu-1}})$.
 \State Let $i_{y_\ell} \in \ITEMS$ with $\ell \in [\mu-1]$ be the
 item such that $y_\ell$ is the parent of $y_\mu$ in 
 $T'$. 
\State \Comment{Note that $\ell < \mu$ for the parent $y_\ell$ of
  $y_\mu$,  
 since
 $y_1,\ldots,y_\mu$ are sorted in document order.}
 \State \Return the first element of the $y_\mu$-list of $i_{y_\ell}$
 \EndFunction
 \State
 \Procedure{visit}{$\mathcal{I}$}
 \State \textbf{Input:} A list of items $\ITEMS=(i_{y_1},\ldots,i_{y_{k}})$.
 \State Output the tuple $(a_1,\ldots,a_k)$, where $a_\ell$ is the
 constant of item 
  $i_{x_\ell}$, for all $\ell\in[k]$.
 \If{every item in $\mathcal{I}$ is the last of its list}
 \State Halt and output the end-of-enumeration message $\EOE$.
 \EndIf
 \State Let $j\in[k]$ be maximal such that $i_{y_j}$ is not the last
 item of its list.
 \For{$\mu = 1$ \textbf{to} $j{-}1$}
 \State $i'_{y_\mu} \gets i_{y_\mu}$
 \EndFor
 \State $i'_{y_j} \gets \nextlistitem{i_{y_j}}$
 \For{$\mu = j{+}1$ \textbf{to} $k$}
 \State $i'_{y_\mu} \gets \textsc{Set}((i'_{y_1},\ldots,i'_{y_{\mu-1}}),\mu)$
 \EndFor
 \State \textsc{visit}$(i'_{y_1},\ldots,i'_{y_{k}})$.
 \EndProcedure
\end{algorithmic}
\end{algorithm}

For every item $i = \ite{v}{\alpha}{a}$ let $\beta^i \isdef \alpha\frac{a}{v}$.
By construction of the algorithm we know that 
for any $\ITEMS=(i_{y_1},\ldots,i_{y_k})$ for which the procedure \textsc{visit}$(\ITEMS)$ is executed by 
Algorithm~\ref{alg:enum}, and 
for any two variables
$y_j,y_k \in V'$ with $j < k$, 
the following is true:

\begin{itemize}
 \item If $y_j \in \pa[y_k)$, then $\assignb^{i_{y_j}} \subseteq \assignb^{i_{y_k}}$.
 \item Otherwise, $\assignb^{i_{y_j}} \cap \assignb^{i_{y_k}} =  \assignb^{i_w}$, where
 $w$ is 
 the lowest common ancestor of $y_j$ and $y_k$ in $T'$.
\end{itemize}
In particular, 
$\assignb_\ITEMS \isdef \bigcup_{v \in V'} \assignb^{i_v}$ is a well-defined function 
$\assignb_\ITEMS\colon\free(\varphi) \to \Dom$. 

The next lemma establishes the correctness of Algorithm~\ref{alg:enum}.

\begin{lemma}\label{lem:enum-correctness-1}\ \hfill
\begin{enumerate}[(a)]
\item\label{lem:enum-correctness-a}
If $\ITEMS=(i_{y_1},\ldots,i_{y_k})$ is a
list of items for which
the procedure \textsc{visit}$(\ITEMS)$ is executed by 
Algorithm~\ref{alg:enum}, then
$\assignb_\ITEMS \in \bigcup_{i\in\startlist}\extensionsetfreeof{i}$.
\item\label{lem:enum-correctness-b}
For every $\beta \in\bigcup_{i\in\startlist}\extensionsetfreeof{i}$ we have
$\beta=\beta_\ITEMS$ for some list of items
$\ITEMS=(i_{y_1},\ldots,i_{y_k})$ for which
the procedure \textsc{visit}$(\ITEMS)$ is executed by 
Algorithm~\ref{alg:enum}.
\item\label{lem:enum-correctness-c}
Algorithm~\ref{alg:enum} does not output duplicates.
\end{enumerate}
\end{lemma}
\begin{proof}
For the proof of \eqref{lem:enum-correctness-a}
note that
by construction of the algorithm we know that any item $i_v$ of $\ITEMS$ appears either in the start-list or 
in the $v$-list of the item $i_u$, where $u$ is the parent node of $v$ in $T'$. Therefore, every item 
in $\ITEMS$ is fit. 
For all $v \in V'$ let $\gamma_v$ be the following inductively defined assignment:
\begin{itemize}
 \item If $v$ is a leaf of $T'$, let $\gamma_v$ be an assignment $\gamma_v \supseteq \assign\frac{a}{v}$ such that
  $\satisfy{\DB}{\gamma_v}{\bigwedge_{\qatom\in\atoms(\varv)}\qatom}$
  where $i_v = \ite{v}{\assign}{a}$. Note that such 
  an assignment exists since $i_v$ is fit. 
 \item For all other nodes $v$ of $T'$ let $\gamma_v \isdef \bigcup_{w \in N(v)} \gamma_w$. Note that 
  by construction of our data structure 
  we know for any two $w,w' \in N(v)$ that $\gamma_{w}(u) =
  \gamma_{w'}(u)$ holds for all $u \in \pa[v]$.
\end{itemize}
Note that $\dom{\gamma_{\vroot}} = \Vars(\varphi) = \bigcup_{\psi \in \atoms(\vroot)} \Vars(\psi)$ and
$\satisfy{\DB}{\gamma_{\vroot}}{\bigwedge_{\qatom\in\atoms({\vroot})}\qatom}$. Thus,
$\gamma_{\vroot} \in \extensionsetof{i_{\vroot}}$. 
It is straightforward to see that $\beta^{i_v} \subseteq \gamma_v$
for every $v\in V'$. Therefore,
\[
  \beta_\ITEMS\ =\ \bigcup_{v\in V'} \beta^{i_v}\ \subseteq\ \bigcup_{v \in V'} \gamma_v\ =\ \gamma_{\vroot} .
\]
Since $\dom{\beta_\ITEMS} = \free(\varphi)$, we have $\beta_\ITEMS \in
\extensionsetfreeof{i_{\vroot}}$.
Furthermore, 
$i_{\vroot} \in \startlist$, and therefore the proof of \eqref{lem:enum-correctness-a} is complete.

For the proof of \eqref{lem:enum-correctness-b} let $\beta \in \extensionsetfreeof{i}$ for an
 item $i \in \startlist$. There exists 
 an assignment 
 $\beta' \supseteq \beta$
 with $\beta' \in \extensionsetof{i}$. We consider for all $v \in V$ the items
 $i_v = \ite{v}{\assign}{a}$ where $\beta' \supseteq
 \assign\frac{a}{v}$. 
 Note that each item $i_v$ is fit. Therefore, the items $i_v$ are either contained in the start-list or
 in the $v$-list of their parents. 
 By induction we show that $\ITEMS\deff (i_{y_1},\ldots,i_{y_k})$ will be considered by the algorithm.
 This holds for $i_{\vroot}$, since the algorithm
 iterates through all elements in the start-list. 
 This establishes the induction base.
 For the induction step let $u,v \in V$ with $u \in N(v)$.
 By the induction hypothesis, $i_v$ will be considered by the algorithm. Therefore,
 the item $i_u$ will be considered by the algorithm, since it is contained in the $u$-list of $i_v$ and
 the algorithm ensures that all elements of the $u$-list of $i_v$ will be considered.
 In summary, $\textsc{visit}(\ITEMS)$ will be called by the
 algorithm. 
 Since $\beta_\ITEMS = \beta$, the proof of
 \eqref{lem:enum-correctness-b} is complete.

For the proof of \eqref{lem:enum-correctness-c}
we define a linear order $\prec$ on all item tuples $\ITEMS$ by
 setting $(i_{y_1},\ldots,i_{y_k})\prec(i_{y_1}',\ldots,i_{y_k}')$ if, and
 only if, there is a $\mu\in[k-1]$ such that $i_{y_1}=i_{y_1}'$,
 \ldots, $i_{y_\mu}=i_{y_\mu}'$, and $i_{y_{\mu+1}}'$ occurs after
 $i_{y_{\mu+1}}$ in $y_{\mu+1}$-list of $i_{y_\mu}$. 
 By the definition of the algorithm, the procedure
 \textsc{visit}$(\ITEMS)$ is called along this linear order.
 To complete the proof, note that every tuple
 $\beta$ that is output by the algorithm uniquely determines an $\ITEMS$ such that $\beta_\ITEMS =
 \beta$ and $\beta$ is reported by call of \textsc{visit}$(\ITEMS)$.
\end{proof}

\subsection{Preprocessing and update}
\label{sec:PreprocessingAndUpdate}

In the preprocessing phase we compute the 
\Querytree 
(by Lemma~\ref{lem:querytree} this can be done in time $\poly(\queryphi)$), 
and we initialise the data structure for the
empty database. Afterwards, we perform 
$\card{\DBstart}$ update steps
to ensure that the data structure represents the initial
database $\DBstart$. By ensuring that the update time is constant, it
follows that the preprocessing time is linear in the size of the
initial database.

To illustrate the
result of an
update step, consider again our database $\DBstart$ from
Example~\ref{upperbound:ex1} and suppose that the tuple $(b,p)$ is
inserted into relation $\relE$. The data structure $\DS_1$ for the
resulting database $\DB_1$ is shown in Figure~\ref{fig:DSone}.
When an update command arrives, we have to modify our data structure accordingly so
that it meets the requirements described in Section~\ref{sec:data_structure}.
For convenience, we summarise the conditions below.
\begin{enumerate}[(a)]
\item \label{item:invariant-a} \existenceConditionText{}
\item \label{item:invariant-b} For every item
  $i=\ite{\varv}{\assign}{a}$ and every $\varu\in N(\varv)$ the list
  $\llist{i}{\varu}$ 
 consists of 
 all fit items
of the form $\ite{\varu}{\extend{\assign}{\varv}{a}}{b}$ that are
present in the data structure.
\item \label{item:invariant-c} The \start-list \startlist  contains
all fit items of the form $\ite{\vroot}{\emptyset}{b}$
that are
present in the data structure.
\item \label{item:invariant-d} All values of $\fitcount{i}$,
  $\listcount{i}{u}$, and $\startcount$ are correct. 
\end{enumerate}

To be able to quickly
decide whether an item meets the requirement \eqref{item:invariant-a}, 
we store
for every $i=\ite{\varv}{\assign}{a}$ and every $\qatom\in
\atoms(v)$ the number $\presentcount{i}{\qatom}$ of expansions
$\beta\supseteq\extend{\assign}{\varv}{a}$ with
$\dom{\beta}=\Vars(\qatom)$ and $\satisfy{\DB}{\beta}{\qatom}$.
Whenever a tuple is inserted into or deleted from a relation, we
increment or decrement these numbers accordingly. The number of values $\presentcount{i}{\qatom}$
that change after an
insertion or deletion of a tuple into an $r$-ary relation $\relR$
is $r$ times the number of occurrences of $\relR$ in $\query$.
Hence these numbers can be updated in time $\poly(\query)$.
An item $i=\ite{\varv}{\assign}{a}$ satisfies condition \eqref{item:invariant-a} if, and
only if, there is a $\qatom\in\atoms(\varv)$ such that $\presentcount{i}{\qatom}>0$.
By using the arrays $\larray{\varv}$ and maintaining these numbers we
can therefore create or remove items in constant time
in order to preserve invariant \eqref{item:invariant-a}.

Now we show how to update the weights $\fitcount{i}$ and the variables
 $\listcount{i}{u}$, $\startcount$ which store
sums of weights of list elements.
The key lemma is the following (where we let products ranging over 
the empty set be $1$).
\begin{lemma}\label{lem:counting_composition}
  For every item $i=\ite{\varv}{\assign}{a}$ in the data structure it holds that
  \begin{align}
    \label{eq:10}
  \fitcount{i}\quad = \quad
    \textstyle\prod_{\qatom\in\atm{\varv}}\presentcount{i}{\qatom}\;\cdot\;\textstyle\prod_{\varu\in
      N(\varv)}\listcount{i}{u}.
  \end{align}
\end{lemma}
\begin{proof}
  Note that for $\qatom\in\atm{\varv}$ we have
  $\presentcount{i}{\qatom}\in\set{0,1}$. Furthermore,
  $\prod_{\qatom\in\atm{\varv}}\presentcount{i}{\qatom} = 1$ if and
  only if $\satisfy{\DB}{\extend{\assign}{\varv}{a}}{\qatom}$ for all
$\qatom\in\atm{\varv}$.
  Now we have two cases. 

  First suppose that $\varv$ is a leaf in $T$.
  In this situation we know that $\atm{\varv}=\atoms(\varv)$ and
  therefore we have to show that $\fitcount{i} = 
  \textstyle\prod_{\qatom\in\atoms(\varv)}\presentcount{i}{\qatom}$,
  which follows immediately from the definition of
  $\extensionsetof{i}$.

  For the second case suppose that 
  $N(\varv)=\set{\varu_1,\ldots,\varu_\ell}$ for $\ell\geq 1$.
  Note that $\atoms(\varv)$ can be decomposed into pairwise disjoint
  sets $\atm{\varv}$, $\atoms(\varu_1)$, \ldots,
  $\atoms(\varu_\ell)$ and that the pairwise intersection of sets
  $U_j\deff\bigcup_{\qatom\in\atoms(\varu_j)}\Vars(\qatom)$ is $\dom{\assign}\cup\set{\varv}$.
  If there is some $\qatom\in\atm{\varv}$ such that
  $\notsatisfy{\DB}{\extend{\assign}{\varv}{a}}{\qatom}$, then
  $\presentcount{i}{\qatom}=0$, and hence
  $\extensionsetof{i}=\emptyset$ and $\fitcount{i}=0$.
  Otherwise, $\presentcount{i}{\qatom}=1$ for all
  $\qatom\in\atm{\varv}$ and we have to show that
  $\fitcount{i}=\prod_{\varu\in N(\varv)}\listcount{i}{u}$.
  Note that for every $\beta\in\extensionsetof{i}$ and every
  $\varu_j\in N(\varv)$ there is a unique $i'\in\llist{i}{\varu_j}$ such
  that
  $\restrict{\beta}{U_j}\in\extensionsetof{i'}$.
  Furthermore, for every choice of $\ell$ assignments
  $\beta_j\in\extensionsetof{i'_j}$ with $i'_j\in\llist{i}{\varu_j}$ for $j\in[\ell]$, there is a
  unique $\beta=\beta_1\cup\cdots\cup\beta_\ell$ such that
  $\beta\in\extensionsetof{i}$.
  It follows that $\fitcount{i} =
    \textstyle\prod_{j\in[\ell]}\big(\sum_{i'\in\llist{i}{\varu_j}}\Setsize{\extensionsetof{i'}}\big)=\textstyle\prod_{\varu\in 
      N(\varv)}\listcount{i}{u}$, which concludes
    the proof.
\end{proof}

We have already argued that upon insertion and deletion of a fact only a constant number of values
$\presentcount{i}{\qatom}$ change.
It follows from Lemma~\ref{lem:counting_composition} that the same
holds true for the values of $\fitcount{i}$,
$\listcount{i}{u}$, and $\startcount$.
Moreover, the lemma enables us to compute the new values bottom up in
time $\poly(\query)$. 

Now we argue how to ensure condition \eqref{item:invariant-b} and
\eqref{item:invariant-c}.
While recomputing the new weights $\fitcount{i}$ we check for every considered item $i$
whether $\fitcount{i}$ becomes zero or non-zero
(recall that $i$ is fit if and only
if $\fitcount{i}>0$) and add or remove $i$ from the
corresponding list.
For every item this takes constant time, as we have random
access to the items and the lists are doubly linked.

To summarise, the update procedure is as follows.
Upon receiving $\Update\;R(b_1,\ldots,b_r)$ we repeat the following
for all atoms $\qatom=R\varz_{1}\cdots\varz_{r}$ of $\queryphi$ that
satisfy $\big(\varz_s=\varz_t
\ \Rightarrow \ b_s=b_t\big)$ for all $s,t\in[r]$.
Let $d=|\set{\varz_1,\ldots,\varz_r}|$, let
$\vroot=\varv_1,\ldots,\varv_d$ be the path in $T$ from $\vroot$ to
the vertex $v_d$ that represents $\qatom$, and 
denote by
$\assign=\Assign{\varv_1,\ldots,\varv_{d}}{a_1,\ldots,a_{d}}$ the
assignment with $\assign(\varz_{j})=b_j$ for all $j\in[r]$.
Then for $j=d,\ldots,1$ and
$i_j\defi\ite{\varv_j}{\Assign{\varv_1,\ldots,\varv_{j-1}}{a_1,\ldots,a_{j-1}}}{a_j}$
we repeat the following steps.
\begin{enumerate}
\item If $\Update=\Insert$, then create $i_j$ (if it not already exists)
  and increment $\presentcount{i_j}{\qatom}$.\\
  Otherwise, decrement $\presentcount{i_j}{\qatom}$.
\item \label{item:compute_fitcount} Let $\fitcount{i_j}_{\text{\upshape old}}\defi \fitcount{i_j}$
  and compute
  $\fitcount{i_j}$
  using Lemma~\ref{lem:counting_composition} as a product of
  $\poly(\query)$ numbers.
\item If $\fitcount{i_j}>0$, then add $i_j$ to
  $\llist{i_{j-1}}{\varv_j}$ (if $j>1$) or  $\startlist$ (if
  $j=1$) 
  (unless it is already present in the according list).\\
  Otherwise, remove $i_j$ from the corresponding list.
\item \label{item:compute_listcount} Compute the new value of $\listcount{i_{j-1}}{\varv_j}$ (if
  $j>1$) or $\startcount$ (if $j=1$) by subtracting $\fitcount{i_j}_{\text{\upshape old}}$
  and adding $\fitcount{i_j}$.
\item If $\Update=\Delete$ and $\presentcount{i_j}{\qatom}=0$ for all
  $\qatom\in\atoms(\varv_j)$, then delete $i_j$ from the data structure.
\end{enumerate}

\subsection{Counting in the presence of quantifiers}
\label{sec:upperbound_counting}

We have already shown that in order to count the number of output tuples in
the case where $\query$ is quantifier-free, it suffices to report the
value of $\startcount$.
If $\query$
is non-Boolean and contains quantified variables
$\xtuple=(\varx_{k+1},\ldots,\varx_m)$,
we maintain the same
numbers $\fitcount{i}$, $\listcount{i}{u}$, $\startcount$ as before
and additionally 
store numbers $\fitcountfree{i}$,
$\listcountfree{i}{u}$, $\startcountfree$ for all present items
$i=\ite{\varv}{\assign}{a}$ with $\varv\in\free(\query)$.
Recall that
 $\extensionsetfreeof{i} \deff \setc{\restrict{\beta}{\free(\query)}}{\beta\in\extensionsetof{i}}$. For an item $i=\ite{\varv}{\assign}{a}$ with $\varv\in\free(\query)$ we let
\begin{align*}
  \fitcountfree{i} &\ \ \deff \ \ \Setsize{\extensionsetfreeof{i}}
  \\
  \listcountfree{i}{\varu} &\ \ \deff \ \ \textstyle\sum_{i'\in\llist{i}{\varu}}
                     \fitcountfree{i'} \qquad\text{for all }\varu\in
                         N(\varv)\cap \free(\varphi)
                     \\
 \startcountfree &\ \ \deff \ \ \textstyle\sum_{i'\in\startlist}
                   \fitcountfree{i'}.
\end{align*}
By definition we have $\setsize{\eval{\query}{\DB}}=\startcountfree$,
and we use this stored value to answer a $\COUNT$ request 
in time $\bigOh(1)$.
To efficiently update the values, we utilise
the following technical lemma, which is similar to Lemma~\ref{lem:counting_composition}. 
\begin{lemma}\label{lem:counting_composition_free}
  For every item $i=\ite{\varv}{\assign}{a}$ with $\varv\in\free(\query)$ it holds that
  \begin{align}
    \label{eq:10}
    \fitcountfree{i}\ \ = \ \
    \begin{cases}
      \ 0 &\text{if }\fitcount{i}=0 \\
      \ \textstyle\prod_{\varu\in
      N(\varv)\cap\free(\query)}\listcountfree{i}{u}&\text{otherwise.}
    \end{cases}
  \end{align}
\end{lemma}

\begin{proof}
  First note that, in general, we have
  $\extensionsetfreeof{i}=\emptyset$ $\iff$
  $\extensionsetof{i}=\emptyset$,
  and hence $\fitcountfree{i}=0$ $\iff$ $\fitcount{i}=0$.  
  If $N(\varv)\cap\free(\query)=\emptyset$, then the lemma holds as
  $\extensionsetfreeof{i}$ is either empty or 
  consists (only) of the assignment $\restrict{\alpha}{\free(\query)}$,
  and hence $\fitcountfree{i}\in\set{0,1}$.  Now suppose
  that 
  $N(\varv)\cap\free(\query)\neq\emptyset$
  and $\fitcount{i}>0$.  Because $i$ is fit and because
  all descendants of quantified
  variables are quantified, it follows
  that for every $\varu\in N(v)\setminus\free(\query)$ we have
  $\satisfy{\DB}{\extend{\assign}{\varv}{a}}{\exists\xtuple
    \textstyle\bigwedge_{\qatom\in\atoms(\varu)}\qatom}$.  Therefore,
  by the same argument as in the proof of
  Lemma~\ref{lem:counting_composition}
  we obtain that
  $\Setsize{\extensionsetfreeof{i}} = \textstyle\prod_{\varu\in
    N(\varv)\cap\free(\query)} \big(\sum_{i'\in\llist{i}{\varu}}
  \Setsize{\extensionsetfreeof{i'}}\big) = \textstyle\prod_{\varu\in
    N(\varv)\cap\free(\query)} \listcountfree{i}{u}$.
\end{proof}

Now we can 
update the numbers $\fitcountfree{i}$,
$\listcountfree{i}{u}$, $\startcountfree$ in the same way as we
have done for $\fitcount{i}$,
$\listcount{i}{u}$, $\startcount$.
In particular, we enrich the update procedure described in Section~\ref{sec:PreprocessingAndUpdate} by
the following two steps \ref{item:compute_fitcount}a and
\ref{item:compute_listcount}a and execute them after
\ref{item:compute_fitcount} and \ref{item:compute_listcount}, respectively.

\begin{itemize}
\item[\ref{item:compute_fitcount}a.] \label{item:compute_fitcountfree}
  If $\varv_j\in\free(\query)$, let
  $\fitcountfree{i_j}_{\text{\upshape old}}\defi \fitcountfree{i_j}$
  and compute $\fitcountfree{i_j}$ using
  Lemma~\ref{lem:counting_composition_free}.
\item[\ref{item:compute_listcount}a.] \label{item:compute_listcountfree}
  If $\varv_j\in\free(\query)$, compute the new value of $\listcountfree{i_{j-1}}{\varv_j}$ (if
  $j>1$) or $\startcountfree$ (if $j=1$) by subtracting
  $\fitcountfree{i_j}_{\text{\upshape old}}$ and adding
  $\fitcountfree{i_j}$.
\end{itemize}

\makeatletter{}%

\section{Discussion}\label{sec:discussion}

We studied the complexity of answering conjunctive queries under
database updates and showed that they can be
evaluated efficiently if they are \qhier.
For \emph{Boolean} conjunctive queries and the task of computing the
result size of non-Boolean queries we proved corresponding lower
bounds based on algorithmic conjectures and obtained a complete picture of the 
queries that 
can be answered efficiently under updates.
Moreover, the \qhier property also precisely characterises those \emph{self-join free}
conjunctive queries that can be enumerated efficiently under updates.

A natural open problem is the missing classification of the
enumeration problem for conjunctive queries that contain self-joins. 
As an intriguing example consider the two CQs
\begin{align*}
  \label{eq:2}
  \query_1(\varx,\vary) &\ \ \deff \ \ \bodyjoin{\relE\varx\varx \und \relE\varx\vary \und
  \relE\vary\vary} \\
  \query_2(\varx,\vary,\varz_1,\varz_2) &\ \ \deff \ \ \bodyjoin{\relE\varx\varx \und \relE\varx\vary \und \relE\vary\vary
  \und \relE\varz_1\varz_2}\,.
\end{align*}
It is easy to see that in the static setting, the results
of both queries can be enumerated with constant delay after linear time
preprocessing (this also immediately follows from 
\cite{Bagan.2007}, since the queries are free-connex acyclic).
But both queries are non-\qhier. 
By similar arguments as in our lower
bound proofs in Section~\ref{sec:lowerbounds}, one can show that 
the results of
$\query_1$ cannot be enumerated with 
$\bigoh(n^{1-\smalleps})$ update time and $\bigoh(n^{1-\smalleps})$
delay, unless the \OMvcon fails (see Appendix~\ref{sec:details_self_joins}).  
However, $\query_2$ \emph{can} be
enumerated with constant delay and constant update time after a linear
time preprocessing phase,
  as the following argument shows.  If
  the $\set{E}$-database $\DB$ is a digraph without loops, then the
  query result is empty.  Otherwise, there is a loop $(c,c)\in E^\DB$
  and we can immediately report the output tuples $(c,c)\times E^\DB$
  with constant delay.  During this enumeration process, which 
  takes time $\Theta(\size{\DB})$, we have enough time to preprocess,
  from scratch, the
  query $\query_1$
  on the database $\DBstrich$ obtained from $\DB$ by deleting the
  tuple $(c,c)$. Afterwards, we can
  enumerate with constant delay the remaining tuples, i.e., the tuples
  in $\eval{\query_1}{\DBstrich}\times \relE^{\DB}$
  (see 
   the Appendix~\ref{sec:details_self_joins} for details). 

Note that even in the static setting, a complexity classification for
enumerating the results of join queries with self-joins
is not in sight and it seems likely that there is no structural
characterisation of tractable queries (see \cite{Bulatov.2012} for a discussion on that
matter).

On a more conceptual level, the main contribution of this paper is to
initiate a systematic theoretical investigation of the 
computational complexity of query evaluation
under database updates.
We are excited by the fruitful connections between database theory and
the theory of dynamic algorithms ---
in particular, that  a known concept from the database
literature that classifies hard queries in various settings (``being
non-hierarchical'') is tightly connected to the underlying
combinatorial hardness shared by many dynamic algorithms (as captured
by the \OMvcon).
We suspect that there are further
settings in which the theory of
dynamic algorithms helps to advance our understanding of query evaluation
under a dynamically changing database.
Currently, we are working towards characterising the complexity of more expressive
queries such conjunctive queries with negation and unions of
conjunctive queries.

\ifthenelse{\isundefined{\USEBIBLATEX}}{
\bibliographystyle{abbrv}
\bibliography{literature}
}{
        \section{Bibliography}
        \printbibliography 
}

{%
\makeatletter{}%

\clearpage
\onecolumn

\appendix
\section*{APPENDIX}

This appendix contains proof details that 
were omitted from Section~\ref{sec:discussion}.

\makeatletter{}%

\bigskip

\section{Details on the Discussion about Queries with Self-Joins}
\label{sec:details_self_joins}

\medskip 

\noindent
We prove the statements about the dynamic enumeration complexity
of the queries $\query_1$ and $\query_2$ as discussed in Section~\ref{sec:discussion}.

\begin{lemma}\label{lem:ExxExyEyy}
Suppose there is an $\epsilon>0$ and a dynamic algorithm with
arbitrary preprocessing time and
$\updatetime=\dimn^{1-\smalleps}$ update time that enumerates 
\,$\query_1(\varx,\vary) \deff \bodyjoin{\relE\varx\varx \und \relE\varx\vary \und
  \relE\vary\vary}$\,
with $\delaytime=\dimn^{1-\smalleps}$ delay on databases whose active
domain has size $n$,
then the \OMvcon fails.
\end{lemma}

\begin{proof}
  We show that a dynamic enumeration algorithm for $\query_1(x,y)$ 
  helps to solve \OuMv in
  time $\bigoh(\dimn^{3-\smalleps})$.
  We start the preprocessing phase of our evaluation algorithm for 
  $\query_1$ with the empty database $\DB=(\relE^{\DB})$
  where $\relE^{\DB}=\emptyset$. As this database has
  constant size, the preprocessing phase finishes in constant time.

  Given an $n\times n$ matrix $M$, we fix $2n$ distinct elements 
  $\setc{\verta_\indi,\vertb_\indi}{\indi \in [\dimn]}$ of $\Dom$, 
  and perform at most $n^2$ update steps to insert
  the tuple $(\verta_\indi,\vertb_\indj)$ into $\relE^{\DB}$,
  for all $(i,j)$ with $M_{i,j}=1$.
  All this is done within time
  $n^2\updatetime=\bigOh(n^{3-\epsilon})$, and
  afterwards $\DB$ is the $\set{\relE}$-db 
  with $\relE^{\DB}=\setc{(\verta_\indi,\vertb_\indj)}{\matM_{\indi,\indj}=1}$.

  When receiving two vectors $\vecu^{\,\indt}$ and $\vecv^{\,\indt}$ in the dynamic
  phase of \OuMv, we insert and delete loops in $\DB$ such that the following
  is true:
  \begin{itemize}
   \item
  $(\verta_\indi,\verta_\indi)\in \relE^\DB$
  \ $\iff$ \ the
  $\indi$-th entry of $\vecu^{\,\indt}$ is $1$,
   \item
  $(\vertb_\indi,\vertb_\indi)\in \relE^\DB$ \ $\iff$ \ the
  $\indi$-th entry of $\vecv^{\,\indt}$ is $1$.
  \end{itemize}
   Now we
  enumerate the result of $\query_1(x,y)=
  \bodyjoin{\relE\varx\varx \und \relE\varx\vary \und
    \relE\vary\vary}$ evaluated on $\DB$
  for $2\dimn{+}1$ steps, and we output $1$ if there was some pair
  $(\verta_\indi,\vertb_\indj)$ in the output, and otherwise we output $0$.  Note
  that such a pair occurs among the first $2\dimn{+}1$ output pairs as
  there are at most $2\dimn$ loops $(\verta_\indi,\verta_\indi)$ and
  $(\vertb_\indj,\vertb_\indj)$.  From the definition of $\DB$ it
  follows that the output agrees with
  $(\vecu^{\,\indt})\trans \matM \vecv^{\,\indt}$.  
  For each $t\in[n]$, all this is done within time 
  $2\dimn\updatetime +
  (2\dimn+1)\delaytime=\bigoh(\dimn^{2-\smalleps})$.
  The overall running time is $\bigoh(\dimn^{3-\smalleps})$.
  This contradicts the \OMvcon{} by Theorem~\ref{thm:OuMv}.
\end{proof}

\medskip

\begin{lemma}\label{lem:ExxExyEyyEz1z2}
  The
  results of the query 
\ $\query_2(\varx,\vary,\varz_1,\varz_2) \deff
\bodyjoin{\relE\varx\varx \und \relE\varx\vary \und \relE\vary\vary
  \und \relE\varz_1\varz_2}$ \ 
  can be enumerated with constant delay and constant update time after a linear time preprocessing phase.
\end{lemma}

\begin{proof}
  We first observe that in the static setting, the results of the
  query
  \,$\query_1(\varx,\vary)\deff 
     \bodyjoin{\relE\varx\varx \und \relE\varx\vary \und \relE\vary\vary}$\, 
  can be
  enumerated with constant delay after $\bigoh(\size{\DB})$ preprocessing (this
  is easy to see and also follows from \cite{Bagan.2007}, since the
  query is free-connex acyclic).

  Our dynamic algorithm for enumerating the results of $\query_2$
  stores a doubly linked list of all elements $\vertc$ such that
  $(\vertc,\vertc)\in \relE^\DB$.
  Furthermore, we store an adjacency matrix of $\relE^{\DB}$ as well
  as a doubly linked list of all tuples in $\relE^\DB$.

  To achieve constant update time, we implement the latter by three
  2-dimensional arrays $\arrayA$, $\arrayB_{next}$, $\arrayB_{prev}$ and two
  tuples $\texttt{first}_{\arrayB}$ and $\texttt{last}_{\arrayB}$, all of which are
  initialised by 0.

  Upon an update command of the form $\Insert\;E(i,j)$, we proceed as
  follows. 
  If $\arrayA[i,j]=1$, we are done.
  Otherwise, we set $\arrayA[i,j]\deff 1$.
  If $\texttt{last}_{\arrayB}=0$, then we set
  $\texttt{first}_{\arrayB}\deff \texttt{last}_{\arrayB}\deff [i,j]$.
  Otherwise, we set $[i',j']\deff\texttt{last}_{\arrayB}$,
  $\arrayB_{prev}[i,j]\deff [i',j']$,
  $\arrayB_{next}[i',j']\deff [i,j]$, and
  $\texttt{last}_{\arrayB}\deff [i,j]$.

  Upon an update command of the form $\Delete\;E(i,j)$, we proceed as
  follows. 
  If $\arrayA[i,j]=0$, we are done. 
  Otherwise, we set $\arrayA[i,j]\deff 0$.
  If $[i,j]=\texttt{first}_{\arrayB}=\texttt{last}_{\arrayB}$, then we
  set $\texttt{first}_{\arrayB}\deff\texttt{last}_{\arrayB}\deff 0$
  and $\arrayB_{next}[i,j]\deff\arrayB_{prev}[i,j]\deff 0$. 
  Otherwise, if $\texttt{first}_{\arrayB}=[i,j]$, then we set 
  $\texttt{first}_{\arrayB}\deff\arrayB_{next}[i,j]$ and
  $\arrayB_{prev}\texttt{first}_{\arrayB}\deff
  \arrayB_{next}[i,j]\deff 0$.
  Otherwise, if $\texttt{last}_{\arrayB}=[i,j]$, then we set 
  $\texttt{last}_{\arrayB}\deff\arrayB_{prev}[i,j]$ and
  $\arrayB_{next}\texttt{last}_{\arrayB}\deff
  \arrayB_{prev}[i,j]\deff 0$.
  Otherwise, we set 
  $\arrayB_{next}\arrayB_{prev}[i,j]\deff \arrayB_{next}[i,j]$,
  $\arrayB_{prev}\arrayB_{next}[i,j]\deff \arrayB_{prev}[i,j]$, and
  $\arrayB_{prev}[i,j]\deff \arrayB_{next}[i,j]\deff 0$.

  Of course, we can use a similar data structure to store, and update
  within constant time, 
  the doubly linked list of all elements $\vertc$ such that
  $(\vertc,\vertc)\in \relE^\DB$.

  Upon a call of the $\ENUMERATE$ routine,
  we let $c_0$ be the first element in the list of all elements $\vertc$ such that
  $(\vertc,\vertc)\in \relE^\DB$.
  If no such element exists, we
  know that the query result is empty, and we can immediately output
  the end-of-enumeration message $\EOE$.
  Otherwise, 
  we immediately start to output all tuples
  $(\vertc_0,\vertc_0)\times \relE^\DB$ with constant delay $\delaytime$ (where we
  choose the constant $\delaytime$ to be large enough).  As this
  takes time $\delaytime\setsize{\relE^\DB}=\Omega(\delaytime\|\DB\|)$, there is
  enough time to perform, in the meantime, the full linear time preprocessing phase of the static
  enumeration algorithm for the query $\query_1(x,y)$ on the database
  $\DBstrich$ that is obtained from $\DB$ by deleting the tuple $(c_0,c_0)$.
  This allows to afterwards
  enumerate with constant delay all remaining tuples in $\query_2(\DB)$, i.e., all tuples
  in $\eval{\query_1}{\DBstrich}\times \relE^{\DB}$.
  \end{proof}

}

\end{document}

%
%
%
%
 %

%
%
%
%
